\documentclass[aos,preprint]{imsart}

\makeatletter
\let\@fnsymbol\@arabic
\makeatother

\RequirePackage[OT1]{fontenc}
\RequirePackage{amsthm,amsmath,amssymb}
\RequirePackage{natbib}
\RequirePackage[colorlinks,citecolor=blue,urlcolor=blue]{hyperref}

\usepackage[pdftex]{graphicx}
\usepackage{array}
\usepackage{algorithm}

\usepackage{appendix}

\usepackage{color}
\usepackage{algpseudocode}
\usepackage{theoremref}
\usepackage{multirow}

\arxiv{arXiv:1410.8260}

\startlocaldefs

\defcitealias{taylor2013tests}{Taylor, Loftus and Tibshirani}
\defcitealias{fithian2014optimal}{Fithian, Sun and Taylor}

\newcommand{\citepryan}{(\citetalias{taylor2013tests}, \citeyear{taylor2013tests})}
\newcommand{\citetryan}{\citetalias{taylor2013tests} (\citeyear{taylor2013tests})}

\newtheorem{lemma}{Lemma}
\newcolumntype{N}{@{}m{0pt}@{}}
\numberwithin{equation}{section}
\theoremstyle{plain}
\newtheorem{thm}{Theorem}[section]
\newtheorem{test}{Test}[section]

\endlocaldefs

\begin{document}

\begin{frontmatter}
\title{Selecting the number of principal components: estimation of the true rank of a noisy matrix}
\runtitle{On Rank Estimation in PCA}

\begin{aug}
\author{\fnms{Yunjin} \snm{Choi}\ead[label=e1]{yunjin@stanford.edu}},
\author{\fnms{Jonathan} \snm{Taylor}\thanksref{t1}\ead[label=e2]{jonathan.taylor@stanford.edu}}
\and
\author{\fnms{Robert} \snm{Tibshirani}\thanksref{t2}
\ead[label=e3]{tibs@stanford.edu}}

\thankstext{t1}{Supported by NSF Grant DMS-12-08857 and AFOSR Grant 113039.}
\thankstext{t2}{Supported by NSF Grant DMS-99-71405 and NIH Contract N01-HV-28183.}

\runauthor{Y. Choi, J. Taylor, and R. Tibshirani}

\affiliation{Stanford University}

\address{Y. Choi\\
Department of Statistics\\
Stanford University\\
Stanford, California 94305 \\
USA\\
\printead{e1}
}

\address{J. Taylor\\
Department of Statistics\\
Stanford University\\
Stanford, California 94305 \\
USA\\
\printead{e2}
}

\address{R. Tibshirani\\
Department of Health, Research \& Policy\\
Department of Statistics\\
Stanford University\\
Stanford, California 94305 \\
USA\\
\printead{e3}
}
\end{aug}

\begin{abstract}
Principal component analysis (PCA) is a well-known tool in multivariate statistics. 
One significant challenge in using PCA is the choice of the number of components. 
In order to address this challenge, we propose an {\em exact} distribution-based method for 
hypothesis testing and construction of confidence intervals for signals in a noisy matrix. Assuming Gaussian noise, we use the conditional distribution of the 
singular values of a Wishart matrix and derive exact hypothesis tests and confidence 
intervals for the true signals. 
Our paper is based on the  approach of \citetryan $\;$  for  testing the global null: we generalize it to  test for any number  of  principal components, and derive an integrated version with greater power. 
In simulation studies we        find     that our proposed methods compare well to existing approaches.
\end{abstract}



\end{frontmatter}

\section{Introduction}
\subsection{Overview}
       Principal component analysis (PCA) is a commonly used method in multivariate statistics. 
        It can be used for a variety of purposes including as a descriptive tool for examining the structure
 of a data matrix, as a pre-processing step for reducing the dimension of the column space of the matrix \citep{josse2012selecting}, or for matrix completion \citep{cai2010singular}. 
 
        One important challenge  in PCA is how to determine the number of components to retain.
  \cite{jolliffe2005principal} provides an excellent summary of existing approaches to determining the number of components, grouping them into three branches:  subjective methods  (e.g., the scree plot), distribution-based test tools (e.g., Bartlett's test), and computational procedures (e.g., cross-validation).
Each branch has advantages as well as disadvantages, and no single method has emerged as the community standard.

        Figure \ref{fig:scor} offers a scree plot as an example. The data are five test scores
from 88 students (taken from \citet{mardia1979multivariate}); the figure shows
 the five singular values in decreasing
order.
\begin{figure}
        \center
            \includegraphics[width=.5\textwidth]{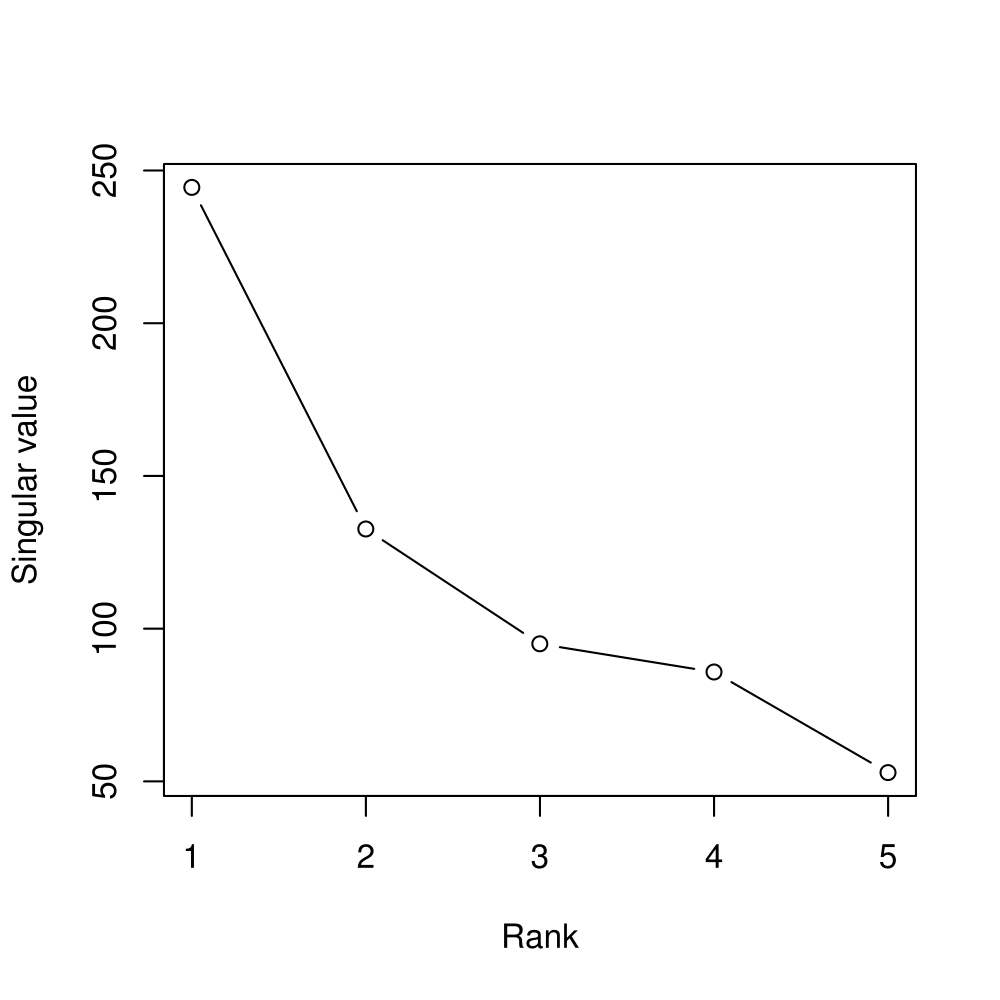} 
       \caption{Singular values of the score data in decreasing order. The data consist of exam scores of 88 students on five different topics (Mechanics, Vectors, Algebra, Analysis and Statistics).}
     \label{fig:scor}
        \end{figure}
The ``elbow'' in this plot seems to occur at rank two or  three, but it is not clearly visible. We revisit this example with our proposed approach in Section \ref{subsec:real_example}.

In this paper, we propose a statistical method for determining the rank of the signal matrix in a noisy matrix model. The estimated rank here corresponds to the number of components to retain in PCA. This method is derived from the conditional  Gaussian-based distribution of the singular values, and yields  exact p-values and confidence intervals.
\subsection{Related work\label{sec:related}} 
        Our method is motivated by the {\em Kac-Rice test} \citepryan, an exact method for testing and constructing confidence intervals for signals under the global null hypothesis in adaptive regression. Our work corresponds to the {Kac-Rice test} under the  global null scenario, in a penalized regression minimizing the Frobenius norm with a nuclear norm penalty. 
        In this paper, we extend the test and confidence intervals to the general case with improved power. The resulting statistic uses the survival function of the conditional distribution of the eigenvalues of a Wishart matrix.
        
        In the context of inference based on  the distribution of eigenvalues, \cite{muirhead} and, more recently, \cite{kritchman2008determining} have proposed methods for testing essentially  the same hypothesis as in this paper. Both \cite{muirhead} and \cite{kritchman2008determining} benefit from an  asymptotic distribution of the test statistic: \cite{muirhead}   forms a likelihood ratio test with the asymptotic Chi-square distribution. \cite{kritchman2008determining} use the Tracy-Widom law, which is the asymptotic distribution of the largest eigenvalue of a Wishart matrix, incorporating the result of \cite{johnstone2001distribution}. We provide a method  for constructing confidence intervals in addition to hypothesis testing, with the procedures being exact.
\subsection{Organization of the paper}
The rest of the paper is organized as follows.
In Section \ref{sec:method}, we propose methods based on the distribution of eigenvalues of a Wishart matrix. Section \ref{sec:kic-rice} introduces the main points of the {\em Kac-Rice test} \citepryan\  from which we derive our proposals. Then, we suggest a procedure for  hypothesis testing of the true rank of a signal matrix in Section \ref{sec:hypothesis}. Our method for constructing exact confidence intervals of signals is described in Section \ref{sec:conf}.

In Section \ref{sec:rank}, we propose a method for estimating the rank of a signal matrix from the hypothesis tests. We illustrate sequential hypothesis testing procedure for determining the matrix rank,  based on a proposal of \citet{g2013false}.

Section \ref{sec:noise} introduces a data-driven method for estimating the noise level. Our suggested method is analogous to cross-validation. 
In order to apply cross-validation to a data matrix, we use a method suggested by \cite{mazumder2010spectral}. 

Section \ref{sec:sim} provides additional examples of the suggested methods.  In Section \ref{subsec:unknownSig}, simulation results with estimated noise level are presented. Section \ref{subsec:nonGaussian} shows simulation results with non-Gaussian noise to check for robustness. 
We revisit the real data example introduced in Figure \ref{fig:scor} in Section \ref{subsec:real_example}. The paper concludes with a brief discussion in Section \ref{sec:conclusion}. 
\section{A Distribution-Based Method\label{sec:method}}
        Throughout this paper, we assume that the  observed data matrix $Y \in \mathbb{R}^{N \times p}$  is the sum of a low-rank signal matrix $B \in \mathbb{R}^{N \times p}$ and a Gaussian noise matrix $E \in \mathbb{R}^{N \times p} $ as follows:
        \begin{eqnarray*}
               &&  Y = B + E, \\
               &&{\rm rank}(B) = \kappa < \min (N, p) \\ 
                && E_{ij} \sim\mbox{ i.i.d.  }  N(0, \sigma^2) \mbox{ for }i \in \{1, \cdots, N\}, \mbox{ } j \in \{1, \cdots, p \}\\
        \end{eqnarray*}
that is, 
        \begin{eqnarray}
                \label{gaussianModel}
                Y \sim N(B, \sigma^2 I_{N} \otimes I_{p}).
        \end{eqnarray}
In this paper, without loss of generality, we assume that $N>p$.
We focus on the estimation of $\kappa$, the rank of the signal matrix $B$, and the construction of confidence intervals for the signals in $B$.

We first review the global null test and confidence interval construction of the first signal of \citetryan \space and its application in matrix denoising problem in Section \ref{sec:kic-rice}.
Then we extend the global null test to a general test procedure for testing $H_{k,0}: {\rm rank}(B) \leq k-1 \mbox{ versus } H_{k,1}: {\rm rank}(B) \geq k \mbox{ for } k = 1, \cdots, p-1$ in Section \ref{sec:hypothesis} and describe how to construct confidence intervals for the $k^{th}$ largest signal parameters in Section \ref{sec:conf}. 
\subsection{Review of the Kac-Rice test\label{sec:kic-rice}}
We briefly discuss the framework of \citetryan \space and its application to a matrix denoising problem, which we extend further later in this paper. Section \ref{subsec:global} covers the testing procedure for a class of null hypothesis. This null hypothesis corresponds to the global null in our matrix denoising problem. In Section \ref{subsec:conf}, we construct a confidence interval for the largest signal.
\subsubsection{Global null hypothesis testing\label{subsec:global}}
         \citetryan \space derived the {\em Kac-Rice test}, an exact test for a class of regularized regression problems of the following form:
        \begin{eqnarray} 
                \label{generalObj}
                \hat{\beta} \in \underset{\beta \in \mathbb{R}^p}{\mbox{argmin}}
                \frac{1}{2} ||y - X\beta||_{2}^{2} + \lambda \cdot \mathcal{P}(\beta)
        \end{eqnarray}
with an outcome $y \in \mathbb{R}^{p}$, a predictor matrix $X \in \mathbb{R}^{N \times p}$ and a penalty term $\mathcal{P}(\cdot)$ with a regularization parameter $\lambda \geq 0$. 
Assuming that the outcome $y \in \mathbb{R}^{p}$ is generated from
        \begin{eqnarray*}
                y \sim N(X\beta, \Sigma),
        \end{eqnarray*}
the {\em{Kac-Rice test}} \citepryan \space provides an exact method for testing
        \begin{eqnarray}
                        \label{eqn:genhyp}
                H_0:  \mathcal{P}(\beta) = 0
        \end{eqnarray}
under the assumption that the penalty function $\mathcal{P}$ is a support function of a convex set $\mathcal{C} \subseteq \mathbb{R}^p$. i.e.,
        \begin{eqnarray*}
              \mathcal{P}(\beta) = \max_{u \in \mathcal{C}} u^{T}\beta.
        \end{eqnarray*}
        
        When applied to a matrix denoising problem of a popular form, (\ref{eqn:genhyp}) becomes a global null hypothesis:
        \begin{eqnarray}
                \label{eqn:globnull}
                H_0 : \Lambda_1 = 0 \equiv {\rm rank}(B)=0 \equiv B=0_{N \times p}
        \end{eqnarray}
where $\Lambda_1 \geq \Lambda_2 \geq \cdots \geq \Lambda_p \geq 0$ denote the singular values of $B$. 
        Here are the details.
    For an observed data matrix $Y \in \mathbb{R}^{N  \times p}$, a widely used method to recover the signal matrix $B$ in (\ref{gaussianModel}) is to solve the following criterion:
        \begin{eqnarray}
                \label{obj}
                \hat{B} \in \underset{B \in \mathbb{R}^{N \times p}}{\mbox{argmin}} 
                \frac{1}{2} ||Y-B||_{F}^2 + \lambda ||B||_{*} \mbox{ where } \lambda > 0
        \end{eqnarray}
where $||\cdot||_F$ and $||\cdot||_{*}$ denote a Frobenius norm and a nuclear norm respectively. The nuclear norm plays an analogous role as an $\ell_1$ penalty term in lasso regression \citep{tibshirani1996regression}. The objective function (\ref{obj}) falls into the class of regression problems described in (\ref{generalObj}), with the predictor matrix $X$ being $I_{N} \otimes I_{p}$  and the penalty function $\mathcal{P}(\cdot)$ being
        \begin{eqnarray*}
                \mathcal{P}(B) = ||B||_{*} = \max_{u \in \mathcal{C}}   \langle u, B \rangle
        \end{eqnarray*}
with $\mathcal{C} = \{ A: ||A||_{op} \leq 1 \}$ where $||\cdot||_{op}$ denotes a spectral norm. We can therefore directly apply the  {\em Kac-Rice test} with the resulting test statistic as follows, under the assumed model  discussed in the beginning of Section \ref{sec:method}:
        \begin{eqnarray}
                \label{testStatistic}
                \mathbb{S}_{1, 0} = \frac{\int^{\infty}_{d_1}  e^{-\frac{z^2}{2\sigma^2}} z^{N-p} \prod_{j=2}^{p} (z^2 - d_j^2)    dz}
                {\int^{\infty}_{d_2} e^{-\frac{z^2}{2\sigma^2}} z^{N-p} \prod_{j=2}^{p} (z^2 - d_j^2)  dz}
        \end{eqnarray}
where $d_1 \geq d_2 \geq \cdots \geq d_p \geq 0$ denote the observed singular values of $Y$.
The test statistic $\mathbb{S}_{1, 0}$ in (\ref{testStatistic}) is uniformly distributed under the null hypothesis (\ref{eqn:globnull}) and provides a p-value for testing the global null hypothesis: the value $\mathbb{S}_{1,0}$ represents the probability of observing more extreme values than $d_1$ under the null hypothesis.

        Viewed differently, the test statistic $\mathbb{S}_{1, 0}$ corresponds to a conditional survival function of the largest observed singular value $d_1$ conditioned on all the other observed singular values $d_2, \cdots, d_p$. The integrand of $\mathbb{S}_{1, 0}$ coincides with the conditional distribution of the largest singular value of a central Wishart matrix up to a constant \citep{james1964distributions}. Its denominator acts as a normalizing constant  because the domain of the largest singular value $d_1$ is $(d_2, \infty)$. A small magnitude of $\mathbb{S}_{1,0}$ implies large $d_1$ compared to $d_2$ and thus supports $H_1: {\Lambda}_1 > 0$.

\subsubsection{Confidence intervals for the largest signal\label{subsec:conf}}
Along with the {\em Kac-Rice test} mentioned in Section \ref{subsec:global}, a procedure for constructing an exact confidence interval for the leading signal in adaptive regression is suggested in \citetryan. As in Section \ref{subsec:global}, by applying the result of \citetryan \space to our matrix denoising setting, we can generate an exact confidence interval for $\tilde{\Lambda}_1$ which is defined as follows:
        \begin{eqnarray*} 
          \tilde{\Lambda}_1 = \langle U_{1}V_{1}^{T}, B\rangle
    \end{eqnarray*}
where $Y = U D V^{T}$ is a singular value decomposition of $Y$ with $D = {\rm diag}(d_1, \cdots, d_p)$ for $d_1  \geq \cdots \geq d_p$, and $U_{1}$ and $V_{1}$ are the first column vectors of $U$ and $V$ respectively. It is desirable to directly find the confidence interval for  $\Lambda_1$ instead of $\tilde{\Lambda}_1$, however, as $B$ is unobservable, $U_1V_1^{T}$ is the ``best guess'' of the unit vector associated with $\Lambda_1$ in its direction.

        To discuss the procedure in detail, in the matrix denoising problem of (\ref{obj}), the result from \citetryan \space yields an exact conditional survival function of $d_1$ around $\tilde{\Lambda}_1$ as follows:
        \begin{eqnarray}
        \label{trueTestStatistic} 
                \mathbb{S}_{1, \tilde{\Lambda}_1} = \frac{\int^{\infty}_{d_1}  e^{-\frac{(z-\tilde{\Lambda}_1)^2}{2\sigma^2}} z^{N-p} \prod_{j=2}^{p} (z^2 - d_j^2)    dz}
                {\int^{\infty}_{d_2} e^{-\frac{(z-\tilde{\Lambda}_1)^2}{2\sigma^2}} z^{N-p} \prod_{j=2}^{p} (z^2 - d_j^2)  dz} \sim \mbox{Unif}(0,1).         
        \end{eqnarray}  
The test statistic $\mathbb{S}_{1,0}$ for testing $H_0: \Lambda_1=0$ in Section \ref{subsec:global} conforms to (\ref{trueTestStatistic}) when it is true that $H_0: \Lambda_1=0$. 
The suggested procedure for constructing the level $\alpha$ confidence interval is as follows:
  \begin{eqnarray} 
        \label{CI}
        CI = \{ \delta:    {\rm min} \left( \mathbb{S}_{1,\delta}, 1-\mathbb{S}_{1,\delta} \right) > \alpha/2 \}.
        \end{eqnarray}
Since $\mathbb{S}_{1,\tilde{\Lambda}_1}$ in (\ref{trueTestStatistic}) is uniformly distributed, we observe that 
        \begin{eqnarray*}
                \mathbb{P}(\tilde{\Lambda}_1 \in CI ) = 1-\alpha
        \end{eqnarray*}
and thus (\ref{CI}) generates an exact level $\alpha$ confidence interval.
\subsection{General hypothesis testing\label{sec:hypothesis}}
In this section, we extend the test for the global null in Section \ref{subsec:global}
to a general test which investigates whether there exists the $k^{th}$ largest signal in $B$. 

Suppose that we want to test the  hypothesis
   \begin{eqnarray}
        \label{generalTest}
        && H_{0,k}: \Lambda_k = 0 \mbox{ versus } H_{1,k}: \Lambda_k > 0 \nonumber \\
        &\Leftrightarrow& H_{0,k}: {\rm rank}(B) < k \mbox{ versus } H_{1,k}: {\rm rank}(B) \geq k.
        \end{eqnarray}
for $k=1, \cdots, p-1$. 
For $k=1$, the null hypothesis in (\ref{generalTest}) corresponds to a global null as in Section \ref{subsec:global}.
With $k=p$ the signal matrix $B$ is full rank under the alternative hypothesis, and the test becomes unidentifiable from a low rank signal matrix with higher noise level problem. Thus, we do not consider the case of $k=p$.

        One of the most straightforward approaches for extending the global test  (\ref{testStatistic}) to testing (\ref{generalTest}) for $k=1,\cdots, p-1$ would be to apply it sequentially.
That is, we can remove the first $k-1$ observed singular values of $Y$
and then apply the test to the remaining $p-k+1$ singular values, in analogy to other methods dealing with essentially the same hypothesis testing \citep{muirhead, kritchman2008determining}.

\begin{figure}[t]
        \center
            \includegraphics[width=\textwidth]{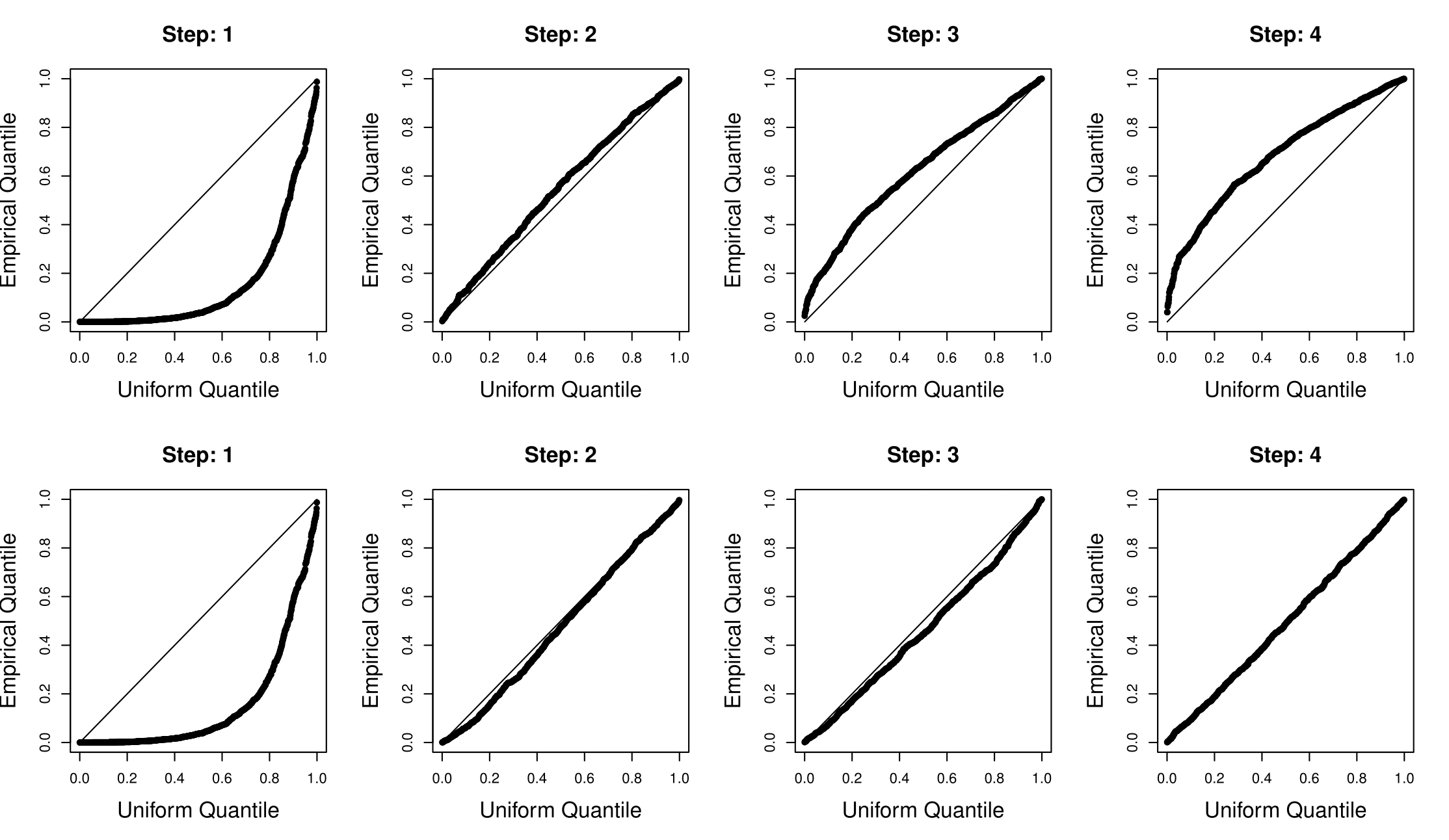}
       \caption{Quantile-quantile plots of the observed quantile of p-values  versus the uniform quantiles. With $N=20$ and $p=10$, the true rank of $B$ is rank$(B)=1$. The top panels are from the sequential {\em Kac-Rice test} and the bottom panels are from the {\tt CSV}. The $k^{th}$ column represents quantile-quantile plots of p-values for testing $H_{0,k}: {\rm rank}(B) \leq k-1$ for steps $k=1,2,3,4$.}
     \label{fig:ex1}
        \end{figure}
How well does this work?
Figure \ref{fig:ex1}
shows an example. Here $N=20$, $p=10$ and there is a rank one signal of moderate
size. The top panels show quantile-quantile plots
of the p-values for the sequential {\em Kac-Rice test} versus the uniform distribution.
We see that the p-values are small when the alternative hypothesis $H_{1,1}$ is true (step 1), and then fairly uniform for testing rank $\leq 1$ versus $> 1$ (step 2) as desired, which is the first case in which the null hypothesis is true.
However the test becomes more and more conservative for higher steps, although the p-values are generated under the null distributions. The conservativeness of these p-values can lead to potential loss of power.
One of the reasons for this conservativeness is that, at step $k=3$, for example, the test does not consider that the two largest singular values have been removed. The test instead plugs the 3rd largest singular value into the place of the 1st largest singular value. As the 1st and the 3rd largest singular values do not have the same distribution, the sequential {\em Kac-Rice test} at $k=3$ is no longer uniformly distributed, and so results in conservative p-values.

        The  plots in the bottom panel come from our proposed {conditional singular value
(CSV)} test, to be described in Section \ref{subsec:csv}. It follows the uniform distribution quite well for all null steps. 
For testing $H_{0,k}: \Lambda_{k}=0$ at the $k^{th}$ step, the {\tt CSV} method takes it into account that our interest is the $k^{th}$ signal by conditioning on the first $k-1$ singular values in its test statistic.

        Section \ref{subsec:csv} presents the {\tt CSV} test procedure in detail. 
        In Section \ref{subsec:integcsv}, we propose an integrated version of the {\tt CSV} which  has better  power.
        Simulation results of the proposed procedures are illustrated in  Section \ref{subsec:hyptest_sim}.
\subsubsection{{The Conditional Singular Value test}\label{subsec:csv}}

        In this section, we introduce a test in  which the test statistic has an ``almost exact'' null distribution under $H_{0,k}$ in (\ref{generalTest}) for $k \in \{ 1, \cdots, p-1\}$. Writing the singular value decomposition of a signal matrix $B$ as $B=U_BD_BV_B^{T}$, and submatrices of a $N \times p$ matrix $X$ by $X = \left[ X^{r} | X^{-r}\right]$ where $X^{r} \in \mathbb{R}^{N \times r}$ and $X^{-r} \in \mathbb{R}^{N \times (p-r)}$, the test of (\ref{generalTest}) can be rewritten as 
        \begin{eqnarray}
                \label{hyp:exact2}
          H_{0,k}: U_{B}^{-(k-1)}U_{B}^{-(k-1)^T} B V_{B}^{-(k-1)}V_{B}^{-(k-1)^T} = 0_{N \times p} \nonumber \\
       \mbox{versus } H_{1,k}: U_{B}^{-(k-1)}U_{B}^{-(k-1)^T} B V_{B}^{-(k-1)}V_{B}^{-(k-1)^T} \neq 0_{N \times p}  
        \end{eqnarray}          
since $U_{B}^{-(k-1)}U_{B}^{-(k-1)^T} B V_{B}^{-(k-1)}V_{B}^{-(k-1)^T} = {\rm diag}(0, \cdots, 0, \Lambda_k, \cdots, \Lambda_p)$.

        As an alternative to (\ref{generalTest}) or (\ref{hyp:exact2}), we derive an exact test statistic for the following hypothesis:
        \begin{eqnarray}
              \label{hyp:exact}
               H_{0,k}: U^{-(k-1)}U^{-(k-1)^T} B V^{-(k-1)}V^{-(k-1)^T} = 0_{N \times p}  \\
       \mbox{versus } H_{1,k}: U^{-(k-1)}U^{-(k-1)^T} B V^{-(k-1)}V^{-(k-1)^T} \neq 0_{N \times p} \nonumber
        \end{eqnarray}
where we write the singular value decomposition of the data matrix $Y$ as $Y=UDV^{T}$. The test (\ref{hyp:exact}) investigates whether there remains signals in the residual space of $Y$ of which the first $k-1$ singular values are removed. Under $H_{0,k}$ of (\ref{hyp:exact2}), when we have strong signals in $B$, $U^{-(k-1)}U^{-(k-1)^T} B V^{-(k-1)}V^{-(k-1)^T}$ approaches $0_{N \times p}$ as in (\ref{hyp:exact}).
 
        The proposed test procedure  is as follows:
        \begin{test}[{Conditional Singular Value test}]
        With a given level $\alpha$, and the following test statistic,
        \begin{eqnarray}
                \label{csv} 
           \mathbb{S}_{k,0} = \frac{\int^{d_{k-1}}_{d_k}  e^{-\frac{z^2}{2\sigma^2}} z^{N-p} \prod_{j \neq k}^{p} |z^2 - d_j^2|    dz}
                {\int^{d_{k-1}}_{d_{k+1}} e^{-\frac{z^2}{2\sigma^2}} z^{N-p} \prod_{j \neq k}^{p} |z^2 - d_j^2|  dz},
        \end{eqnarray}
       where $d_0=\infty$, we reject $H_{0,k}$ if $\mathbb{S}_{k,0} \leq \alpha$ and accept $H_{0,k}$ otherwise.
    \end{test}
Analogous to (\ref{testStatistic}), $\mathbb{S}_{k, 0}$ plays the role of a p-value. It  compares the relative size of $d_k$ ranging between $(d_{k+1}, d_{k-1})$, and a small value of $\mathbb{S}_{k, 0}$ implies a  large value of $d_k$,  supporting the hypothesis $H_{1,k}: \Lambda_{k} > 0$.
We refer to this procedure as the  {\tt conditional singular value test (\tt CSV)}.   
Theorem \ref{lem:csv} shows that this test is exact under (\ref{hyp:exact}). The test statistic $\mathbb{S}_{k,0}$ is a conditional survival function of the $k^{th}$ singular value under $H_{0,k}$: the probability of observing larger values of the $k^{th}$ singular value than the actually observed $d_k$, given $U^{(k-1)}, V^{(k-1)}$ and all the other singular values.
The proofs of this and other results are given in the Appendix.
        \begin{thm}
                \label{lem:csv} 
                  \begin{eqnarray*}
                {\it If} \; U^{-(k-1)}U^{-(k-1) ^T} B V^{-(k-1)}V^{-(k-1)^T} &=& 0_{N \times p},
                 \end{eqnarray*}
  \hskip 1.5in   {\it then}  \; $\mathbb{S}_{k,0} \sim {\rm Unif}(0,1).$

        \end{thm}
        The bottom panels of Figure  \ref{fig:ex1} confirm the claimed Type I error
property of the procedure. After the true rank of one, the p-values are all close to uniform.
\subsubsection{{The Integrated Conditional Singular Value} test\label{subsec:integcsv}}
        As a potential improvement of the {\tt CSV}, we introduce an integrated version of $\mathbb{S}_{k,0}$. Our aim is to achieve  higher power in detecting signals in $B$ compared to the ordinary {\tt CSV}.  Here we integrate out $\mathbb{S}_{k,0}$ with respect to $p-k$ small singular values ($d_{k+1}, \cdots, d_{p}$). The resulting statistic becomes a function of 
$d_1, \cdots, d_k$, only the first $k$ singular values of $Y$,
while the ordinary {\tt CSV} test statistic $\mathbb{S}_{k, 0}$ is a function of all the singular values of $Y$.  The idea is that conditioning on less can lead  to greater power.
        The suggested test statistic is as follows:
                \begin{eqnarray}
                \label{integ_csv}
                \mathbb{V}_{k,0} = \frac{\int^{d_{k-1}}_{d_{k}}  g(y_{k}; d_1, \cdots, d_{k-1})    dy_{k}}
                {\int^{d_{k-1}}_{0}   g(y_{k}; d_1, \cdots, d_{k-1})  dy_{k}},
        \end{eqnarray}
where
\begin{footnotesize}
        \begin{eqnarray}
        \label{gfun}
		g(y_{k};d_1, \cdots, d_{k-1}) = \int \cdots \int \prod_{i=k}^{p}\left( e^{-\frac{y_{i}^2}{2 \sigma^2}} y_{i}^{N-p}\right) \left(\prod_{i=k}^{p}\prod_{j>i} (y_{i}^2 - y_{j}^2)\right)  \\
		\cdot \left(\prod_{i=1}^{k-1} \prod_{j=k}^{p} (d_{i}^2 -y_{j}^2)\right) 1_{ \{ 0 \leq y_p \leq y_{p-1} \leq \cdots \leq y_{k} \leq d_{k-1} \} }  d{y_{k+1}}\cdots d{y_{p}}. \nonumber
	\end{eqnarray}
        \end{footnotesize}
        Our proposed {\tt Integrated Conditional Singular Value (ICSV)} test  is as follows:
        \begin{test}[{\tt ICSV test}]
        With a given level $\alpha$, we reject $H_{0,k}$ if $\mathbb{V}_{k,0} \leq \alpha$ and accept $H_{0,k}$ otherwise, where $\mathbb{V}_{k,0}$ is as defined in (\ref{integ_csv}).
        \end{test}
        As in the ${\tt CSV}$ test, $\mathbb{V}_{k,0}$ works as a p-value for the test, and also it is an exact test for (\ref{hyp:exact}), as is shown in Theorem \ref{lem:integrated_csv}. It is a survival function of the $k^{th}$ singular value given $U^{(k-1)}, V^{(k-1)}$ and the first $k-1$ singular values under $H_{0,k}$. 
        \begin{thm} \
                \label{lem:integrated_csv} 
                  \begin{eqnarray*}
                {\it If} \; U^{-(k-1)}U^{-(k-1) ^T} B V^{-(k-1)}V^{-(k-1)^T} &=& 0_{N \times p},
                 \end{eqnarray*}
  \hskip 1.5in   {\it then}  \; $\mathbb{V}_{k,0} \sim {\rm Unif}(0,1)$

        \end{thm}
Figure \ref{fig:pow} demonstrates that  the {\tt ICSV} procedure achieves higher power than the ordinary {\tt CSV}.

In this paper, we use importance sampling to evaluate the integral in (\ref{integ_csv}) with samples drawn from the eigenvalues of a $(N-k+1) \times (p-k+1)$ Wishart matrix. As the computational cost increases sharply with large $p$,  we are currently unable to compute this test for $p$ beyond say 30 or 40. \  An interesting open problem is the numerical approximation of this integral, in order to scale the test to larger problems. We leave this as future work.
\subsubsection{Simulation Examples\label{subsec:hyptest_sim}}
       In this section, we present results of the {\tt CSV} and the {\tt ICSV} on simulated examples. We compare the performance of these proposed methods with those in \cite{kritchman2008determining} and \cite{muirhead}[Theorem 9.6.2] mentioned in Section \ref{sec:related}, which we refer as the {\em pseudorank} and the {\em Muirhead's method} respectively:
        \begin{test}[{\em Pseudorank}]
        With a given level $\alpha$, and following $\mu_{N,p}$ and $\sigma_{N,p}$,
         \begin{eqnarray*}
        \mu_{N,p} &=& \left(\sqrt{N-\frac{1}{2}} + \sqrt{p-\frac{1}{2}} \right)^{2} \\
        \sigma_{N,p} &=& \left(\sqrt{N-\frac{1}{2}} + \sqrt{p-\frac{1}{2}} \right)\left(\frac{1}{\sqrt{N-1/2}} + \frac{1}{\sqrt{p-1/2}} \right)^{1/3},
        \end{eqnarray*}
  we reject $H_{0,k}$ if 
     \begin{eqnarray*}
        \frac{d_k^2 - \mu_{N,p-k}}{\sigma_{N,p-k}} > s(\alpha)
        \end{eqnarray*}
    where $s(\alpha)$ is the upper $\alpha$-quantile of the Tracy-Widom distribution. 
        \end{test}
        \begin{test}[{\em Muirhead's method}]
        With a given level $\alpha$, and $V_k$ defined as
        \begin{eqnarray*}
        V_k = \frac{(N-1)^{q-1}\prod_{i=k}^p d_{i}^2}{\left(\frac{1}{q} \sum_{i=k}^p d_i^2 \right)^{q}},
        \end{eqnarray*}
        we reject $H_{0,k}$ if 
        \begin{eqnarray*}
        -\left(N-k-\frac{2q^2+q+2}{6q} + \sum_{i=1}^{k-1} \frac{\bar{l}_{q}^2}{ \left( d_i^2 - \bar{l}_{q} \right)^2 } \right) \log V_k > \chi^2_{(q+2)(q-1)/2}(\alpha)
        \end{eqnarray*}
        where $q=p-k+1$, $\bar{l}_q = \sum_{i=k}^{p} d_i^2 / q$ and $\chi^2_m(\alpha)$ denotes the upper $\alpha$ quantile of the $\chi^2$ distribution with degree $m$.
        \end{test}
        We investigate cases with i.i.d Gaussian noise entries with $\sigma^2=1$.
        The observed data matrix $Y^{N \times p}$ has the signal matrix $B$ formed as follows:
        \begin{eqnarray}
        \label{simsetting}
        B = U_B D_{B} V_{B}^T\mbox{ ,  } \Lambda_{i} = m \cdot i \cdot \sigma \sqrt[4]{Np} \cdot I_{ \{ i \leq {\rm rank}(B)\} } 
        \end{eqnarray}
where $D_{B} = {\rm diag}(\Lambda_1, \cdots, \Lambda_p)$ with $\Lambda_1 \geq \cdots \geq \Lambda_p$, and $U_B$, $V_B$ are rotation operators generated from a singular value decomposition of $N \times p$ random Gaussian matrix with $i.i.d.$ entries. 
The signals of $B$   increase linearly.
The constant $m$ determines the magnitude of the signals. From  $m=1$, a {\em phase transition} phenomenon is observed when ${\rm rank}(B)=1$ in which the expectation of the largest singular value of $Y$ starts to reflect the signal \citep{nadler2008finite}.

 We illustrate two cases of $p=10$ and $p=30$ with $N$ fixed to $N=50$. For both cases, we set $m=1.5$ and $\rm{rank}(B)=0,1,2,3$. For $p=10$ and $p=30$, we evaluate the procedure through $3000$ and $1000$ repetitions respectively.
 The true value of the noise level $\sigma^2=1$ is used for all testing procedures.
 
        Figures \ref{fig:pow} and \ref{fig:pow_large} present quantile-quantile plots of the expected (uniform) quantiles versus  the observed quantiles of p-values.
         Under $H_{1, k}$, the {\tt ICSV} test shows improved power compared to the {\tt CSV}, and  close to that of {\em pseudorank}. Both the {\tt CSV} and the {\tt ICSV} show stronger power than {\em Muirhead's test}.
         Under $H_{0,k}$, both the {\tt CSV} and the {\tt ICSV} quantiles nearly agree with the expected quantiles and provide almost exact p-value, as the theory predicts. {\em Pseudorank} estimation becomes strongly conservative for further steps, and the results of {\em Muirhead's test} depend on the size of $N$ and $p$. 
\begin{figure}[tp]
        \center
        \begin{tabular}{c}
                \includegraphics[height=1.5in, width=5in]{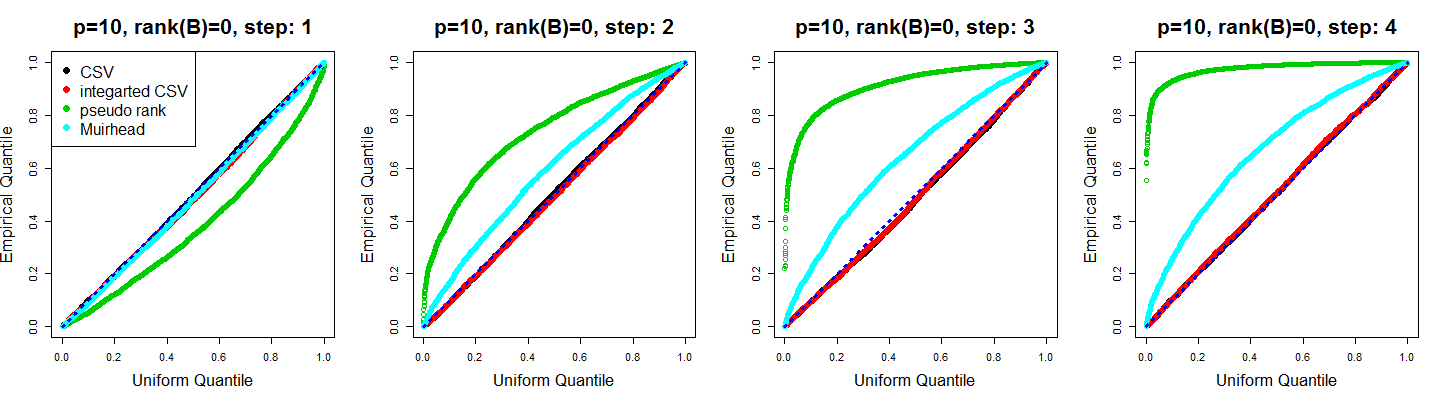} \\
                \includegraphics[height=1.5in, width=5in]{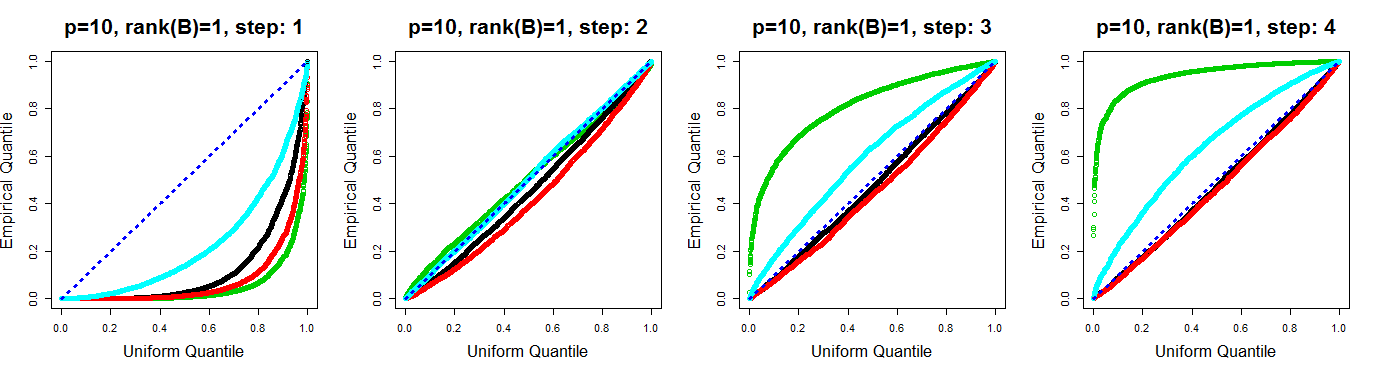} \\
                \includegraphics[height=1.5in, width=5in]{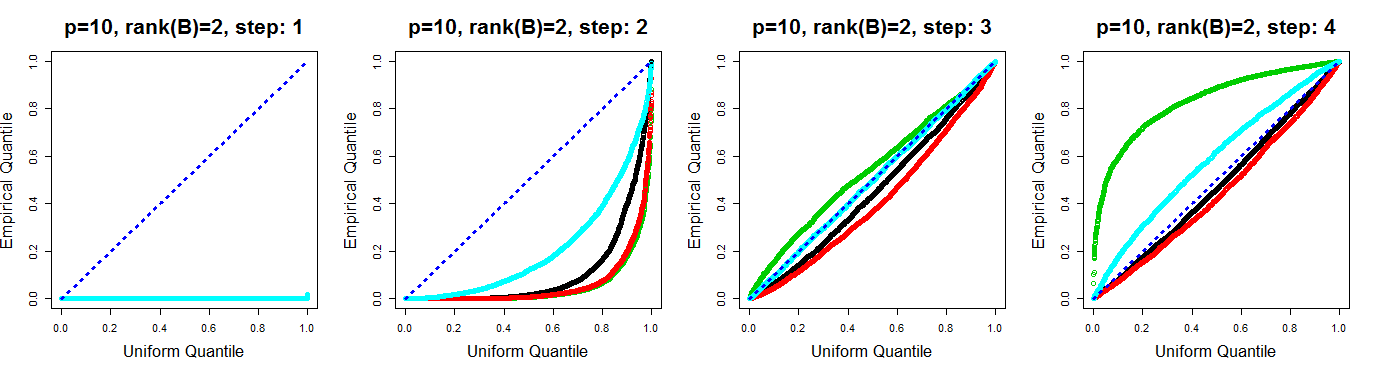} \\
                \includegraphics[height=1.5in, width=5in]{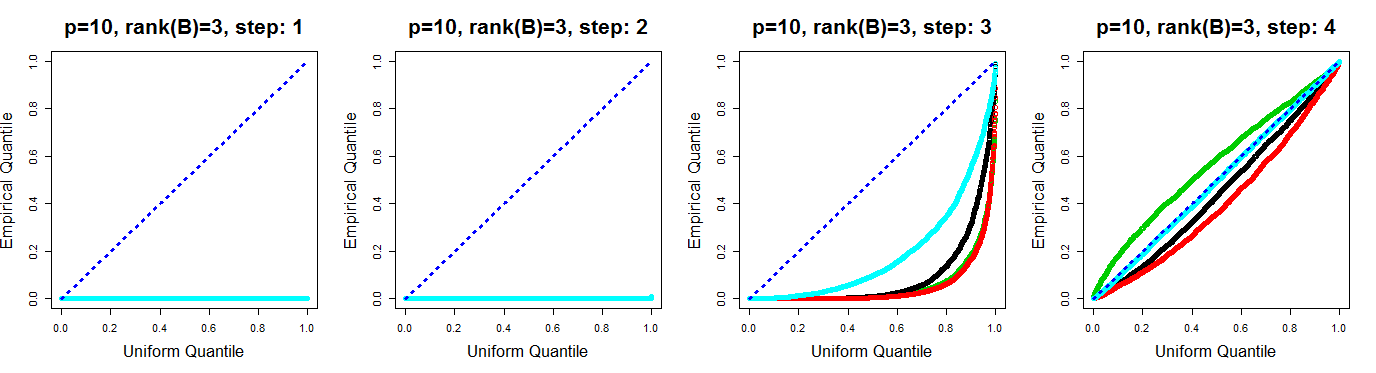} \\
        \end{tabular}
        \caption{Quantile-quantile plots of the  empirical quantile of p-values versus the uniform quantiles when $p=10$ at $m=1.5$. From the top to the bottom, each row represents the case of the true ${\rm rank}(B)$ from  0 to 3. The columns represent the results for testing $H_{0,1}$ to $H_{0,4}$ from the left to the right. \label{fig:pow}
        }
\end{figure}
\begin{figure}[tp]
        \center
        \begin{tabular}{c}
                \includegraphics[height=1.5in, width=5in]{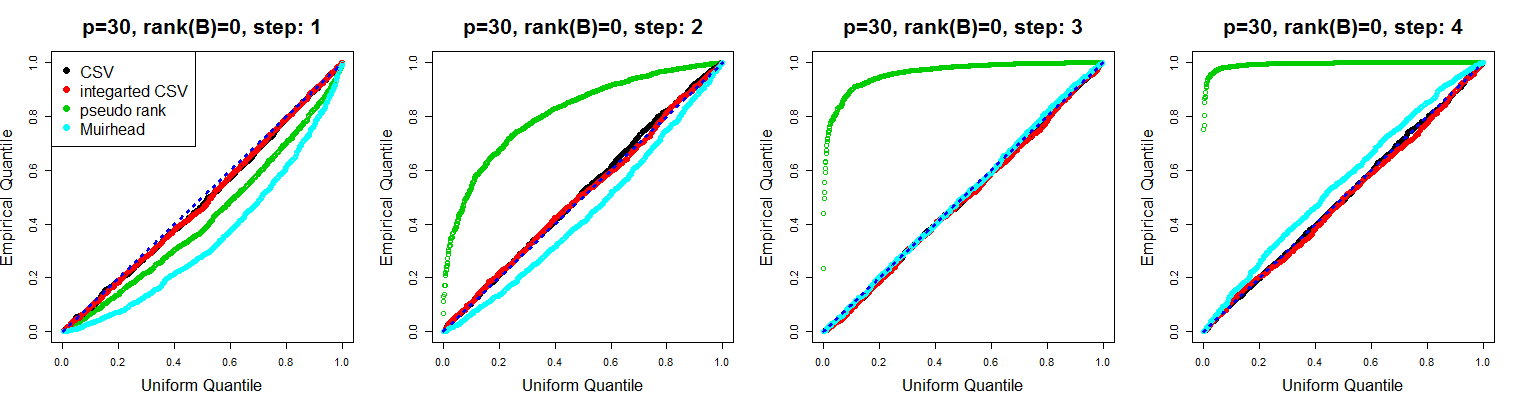} \\
                \includegraphics[height=1.5in, width=5in]{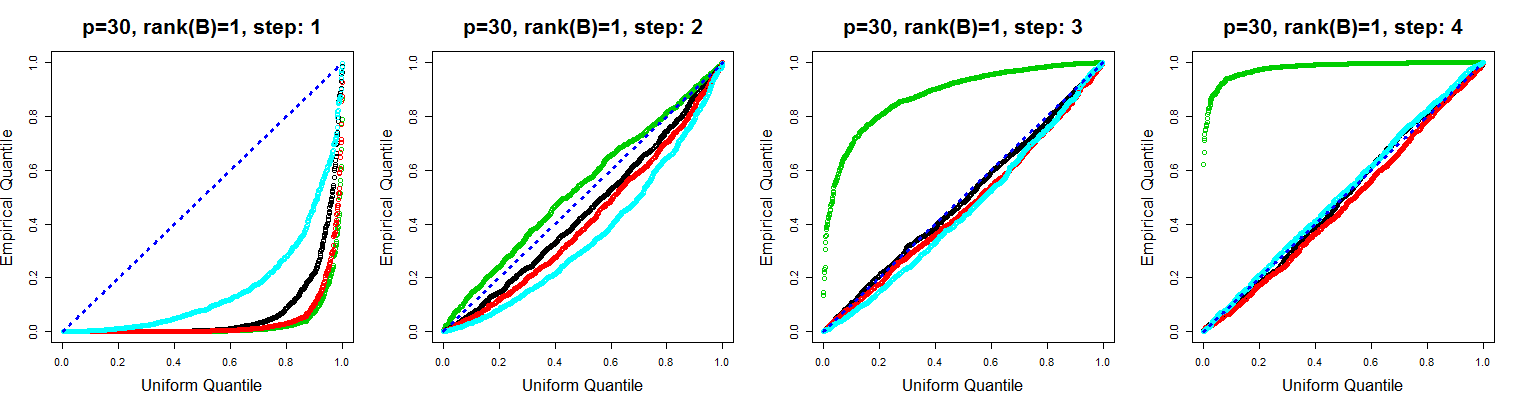} \\
                \includegraphics[height=1.5in, width=5in]{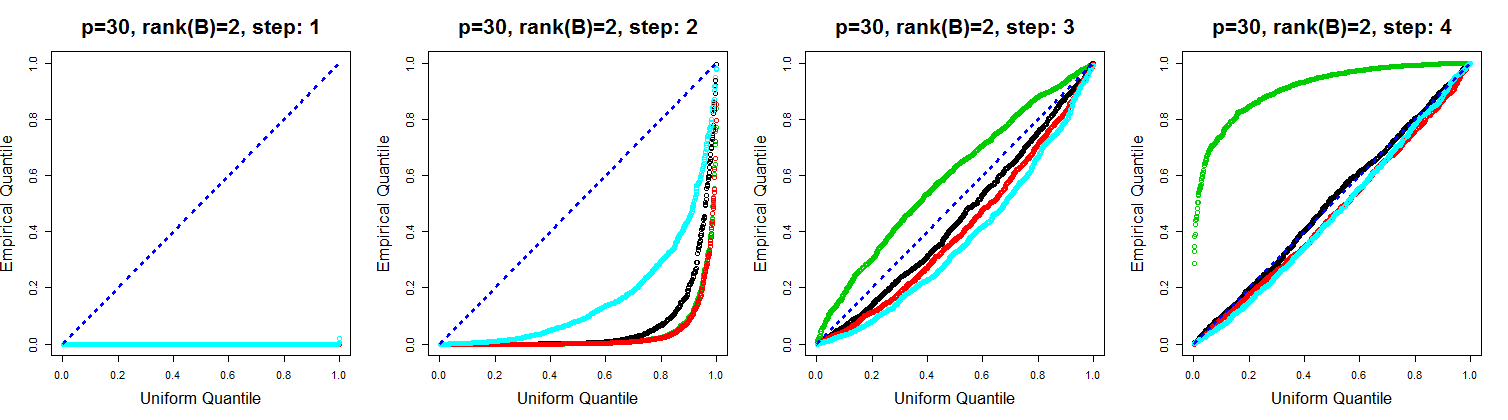} \\
                \includegraphics[height=1.5in, width=5in]{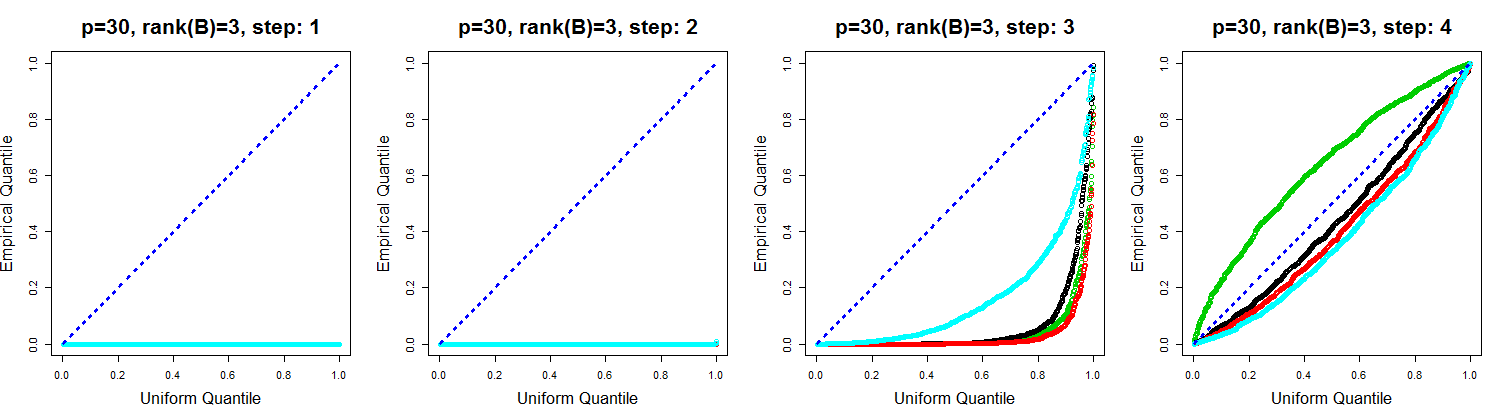} \\
        \end{tabular}
        \caption{Quantile-quantile plots of the  empirical quantile of p-values versus the uniform quantiles when $p=30$ at $m=1.5$. From the top to the bottom, each row represents the case of the true ${\rm rank}(B)$ from 0 to 3. The columns represent the results for testing $H_{0,1}$ to $H_{0,4}$ from the left to the right. \label{fig:pow_large}
        }
\end{figure}
\subsection{Confidence interval construction\label{sec:conf}}
        Here we generalize the exact confidence interval construction procedure of the largest singular value in (\ref{CI}) to the $k^{th}$ signal parameter for any $k=1, \cdots, p-1$.
        
We define the $k^{th}$ signal parameter $\tilde{\Lambda}_k$ as follows:
        \begin{eqnarray}
        \label{tildeLambda_k}
        \tilde{\Lambda}_{k} = \langle U_{k}V_{k}^{T}, B\rangle
        \end{eqnarray}
where $U_{k}$ and $V_{k}$ are the $k^{th}$ column vector of $U$ and $V$ respectively. We propose an approach to construct an exact level $\alpha$ confidence interval of $\tilde{\Lambda}_{k}$.
Our proposed procedure is as follows:
        \begin{eqnarray}
        \label{CI_k}
        CI_{k}(\mathbb{S}) &=& \{ \delta: \min\left( \mathbb{S}_{k,\delta}, 1-\mathbb{S}_{k,\delta} \right) > \alpha/2 \},
        \end{eqnarray}
where
        \begin{eqnarray*}
        \mathbb{S}_{k, \delta} &=& \frac{\int^{d_{k-1}}_{d_k}  e^{-\frac{(z-\delta)^2}{2\sigma^2}} z^{N-p} \prod_{j \neq k}^{p} |z^2 - d_j^2|    dz}
                {\int^{d_{k-1}}_{d_{k+1}} e^{-\frac{(z-\delta)^2}{2\sigma^2}} z^{N-p} \prod_{j \neq k}^{p} |z^2 - d_j^2|  dz}.
        \end{eqnarray*}
We can find the boundary points of $CI_k(\mathbb{S})$ using bisection.
Theorem \ref{exactCI} below shows that $CI_{k}(\mathbb{S})$ is an exact level $\alpha$ confidence interval. This procedure addresses  the general case of (\ref{csv}): $\mathbb{S}_{k,0}$ tests whether  $\tilde{\Lambda}_{k} = 0$.
\begin{thm}
        \label{exactCI}
        $\mathbb{S}_{k, \delta}$ is uniformly distributed when $\delta = \tilde{\Lambda}_k$.
\end{thm}
        
                Figure \ref{fig:coverage} shows that the coverage rate of $CI_k(\mathbb{S})$. As expected,  the coverage rate of the true parameter is close to the target $1-\alpha$.
\subsubsection{Simulation studies of the confidence interval construction}
        We illustrate the coverage rates of $CI_{k}(\mathbb{S})$ in (\ref{CI_k}) on simulated data. The simulation settings are the same as in Section \ref{subsec:hyptest_sim} with $N=50$ and $p=10$. Figure \ref{fig:coverage} shows the coverage rate of the $95 \%$ confidence intervals for the first two signal parameters $\tilde{\Lambda}_1$ and $\tilde{\Lambda}_2$ in (\ref{tildeLambda_k}). 
Here, we vary $m$ in (\ref{simsetting}) from $0$ to $2$.
Large $m$ leads to large magnitude of true $\tilde{\Lambda}_1$ and $\tilde{\Lambda}_2$.
 
        Figure \ref{fig:coverage} shows that regardless of $k=1,2$, true ${\rm rank}(B)$, or the size of $\tilde{\Lambda}_k$, the constructed confidence intervals cover the parameters at the desired level.
\begin{figure}[tp]
\begin{center}
        \begin{tabular}{c}
                \includegraphics[scale=0.15]{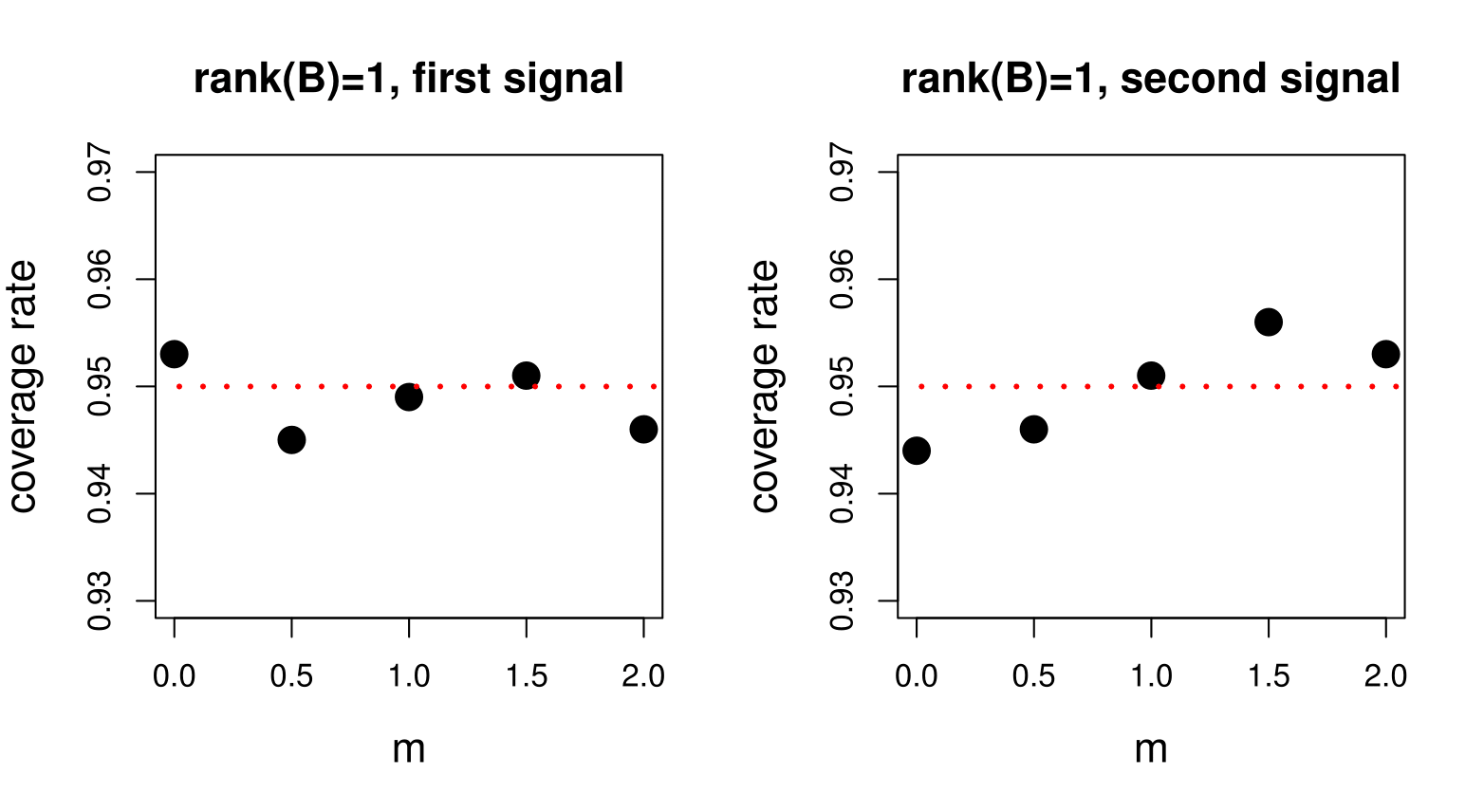} \\
                \includegraphics[scale=0.15]{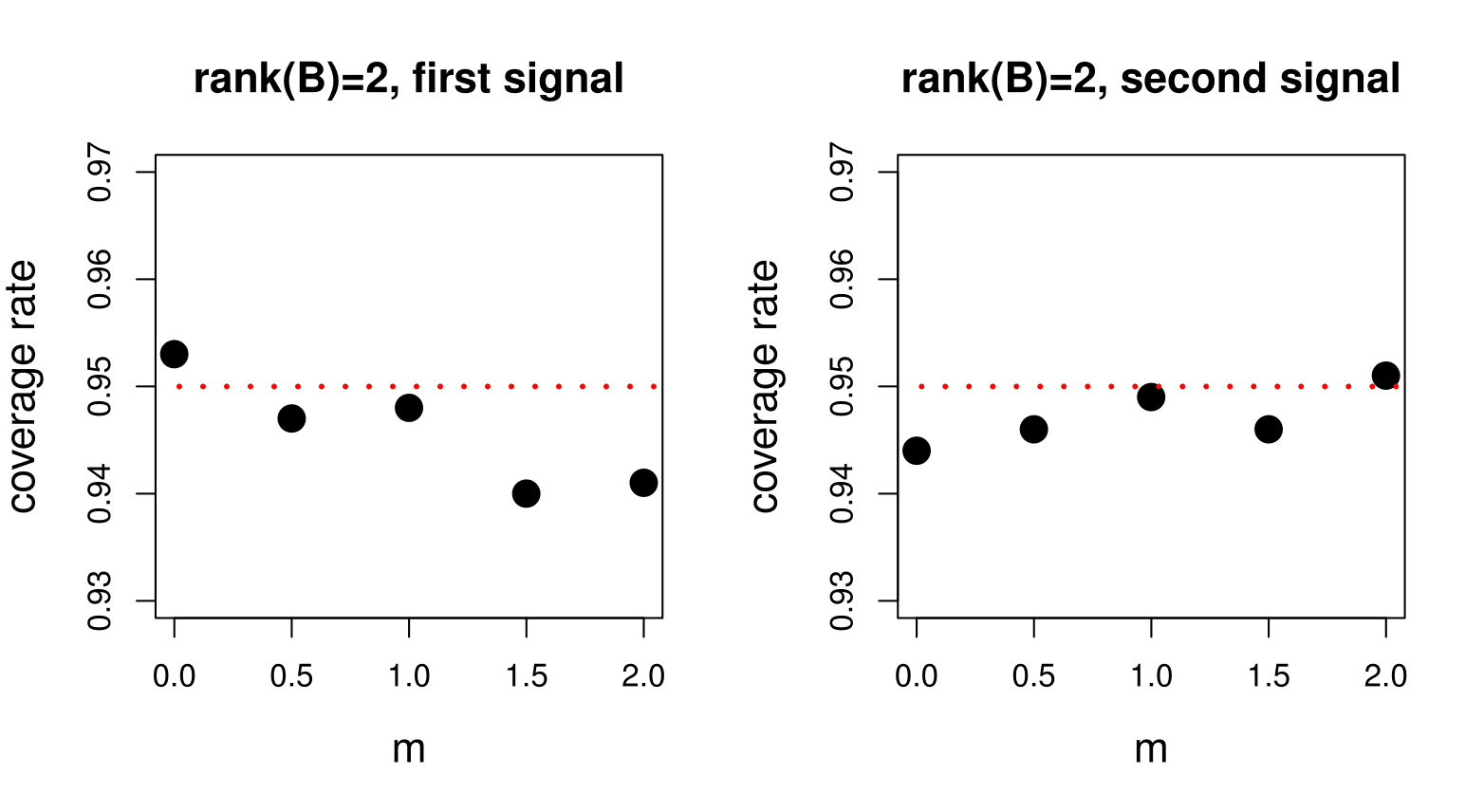} \\
                \includegraphics[scale=0.15]{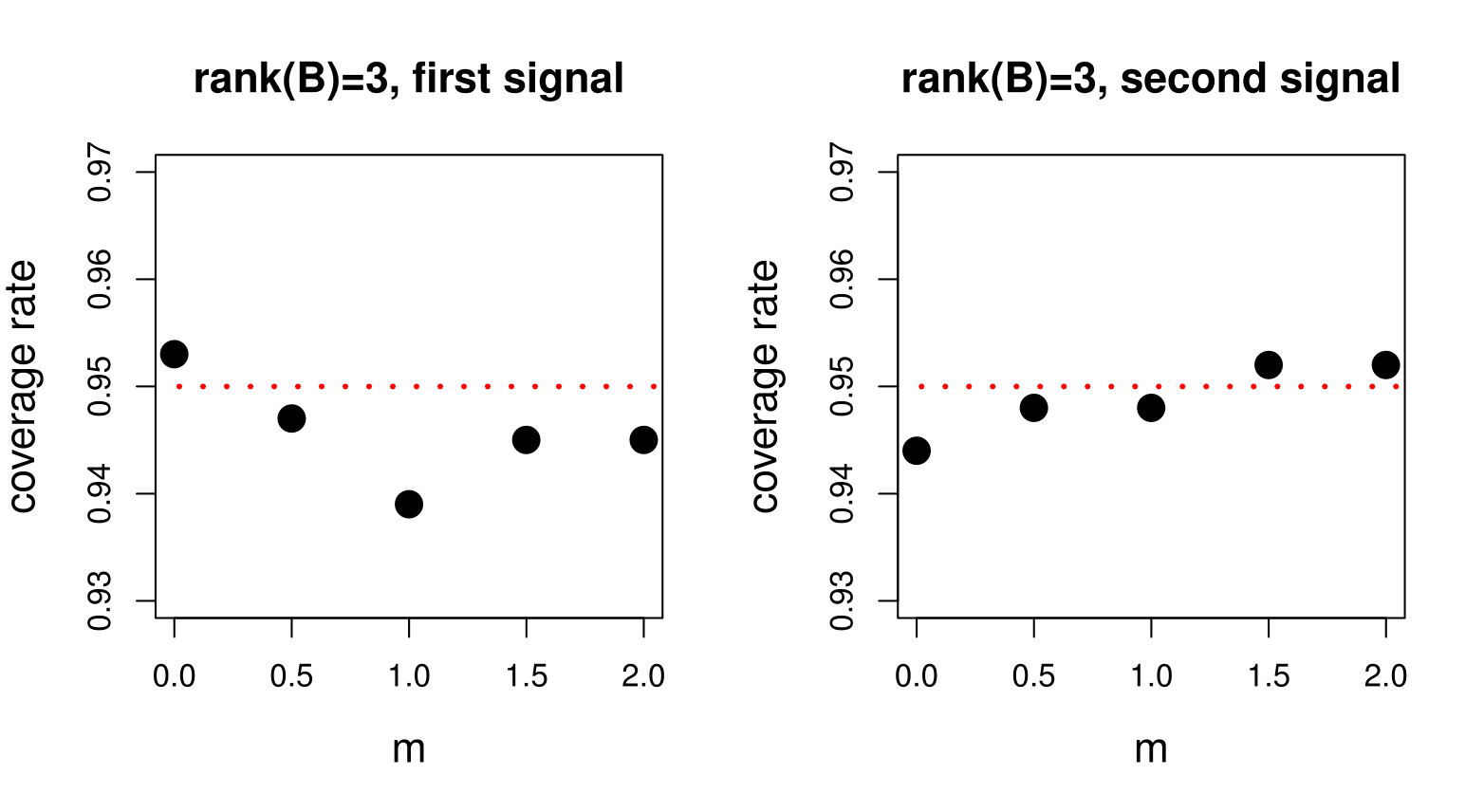} \\
        \end{tabular}
        \end{center}
        \caption{Coverage rate of $\tilde{\Lambda}_1$ and $\tilde{\Lambda}_2$ versus the magnitude $m$ in (\ref{simsetting}) with $N=50$ and $p=10$. The first and the second signal correspond to $\tilde{\Lambda}_1$ and $\tilde{\Lambda}_2$ respectively. Coverage rate denotes the proportion of times that the constructed confidence intervals from $CI_k(\mathbb{S})$ covered the true parameter. From the top to the bottom, each row represents the case of the true ${\rm rank}(B)$ being $0$ to $3$. The columns illustrate the results regarding $\tilde{\Lambda}_1$ and $\tilde{\Lambda}_2$ from the left to the right.\label{fig:coverage}}
\end{figure}
\section{Rank Estimation\label{sec:rank}}
        This section discusses the  selection of the number of principal components, or equivalently, the estimation of the true rank of $B$ under our model assumptions. For determining the rank of $B$, 
        here we investigate the {\em{StrongStop}} procedure \citep{g2013false}, applying it to the tests developed in Section \ref{sec:hypothesis} ({CSV} and {ICSV}).
 
        We determine the rank of $B$ based on our testing results on $H_{0,k}$, since $H_{0,k}$ explicitly tests the range of the true rank. Given the sequence of hypothesis $H_{0,k}:{{\rm rank}(B)\leq k-1}$ with $k=1, \cdots, p-1$, the rejection of these must be carried out in a sequential fashion such that once $H_{k,0}$ is rejected, all $H_{\gamma,0}$ for $\gamma \leq k$ should be rejected as well. Under such sequential testing framework of this kind, it is natural to choose ${\rm rank}(B)$ to be the largest $k$ that rejects $H_{0,k}$. The question here is how to choose the `stopping point' for rejection. 
        
        One of the simplest methods  is to choose the value  $k$ at which $H_{0,k}$ is rejected for the last time with a given level $\alpha$:
        \begin{eqnarray*}
                \hat{\kappa}_{simple} = \max \{k\in \{ 1, \cdots, p-1\}: p_k \leq \alpha \},
        \end{eqnarray*}
which we refer as {\em SimpleStop}.

In this paper, instead of {\em SimpleStop}, we use {\em StrongStop}.
This  procedure  takes sequential p-values as its input and controls family-wise error rate. When the p-values of the sequential tests are uniformly and independently distributed under the null, the {\em StrongStop} procedure controls the family-wise error rate under a given level of $\alpha$ \citep{g2013false}[Theorem 3]. 
        For  rank determination, by the nature of our hypothesis, $H_{0,\gamma}:{\rm rank}(B) \leq \gamma-1$ being true implies $H_{0,k}:{\rm rank}(B) \leq k-1$ also being true for all $k > \gamma$. The family-wise error rate control property in rank determination, therefore, becomes control of rank over-estimation with level $\alpha$ as follows:
        \begin{eqnarray*}
                \mathbb{P}( \hat{\kappa} > \kappa ) \leq \alpha
        \end{eqnarray*}
where $\hat{\kappa}$ denotes the selected ${\rm rank}(B)=\kappa$. The resulting procedure is as follows:
        \begin{eqnarray*}
        \hat{\kappa} = \max \left\lbrace k \in \{ 1, \cdots, p-1 \}: {\rm exp}\left( \sum_{j=k}^{p-1}\frac{\log p_j}{j}\right) \leq \frac{\alpha k}{p-1}\right\rbrace.
        \end{eqnarray*}
Here, $p_k$ denotes the value of either $\mathbb{S}_{k,0}$ of the {\tt CSV} or $\mathbb{V}_{k,0}$ of the {\tt ICSV}, and conventionally $\max ({\emptyset})=0$. The independence of p-values  from our  proposed
testing procedures for $H_{0,k}$ with $k > \kappa$ has yet been established; however {\em StrongStop} shows good performance for strong signals on simulated data.

        Figure \ref{fig:corr} illustrates the simulation results of {\em StrongStop} on p-values from the {\tt ICSV} procedure with $\alpha = 0.1$. The simulation setting is the same as in Section \ref{subsec:hyptest_sim} with $N=50$ and $p=10$. We compare {\em StrongStop} with {\em SimpleStop}.
From Figure \ref{fig:corr}, we observe that for smaller signals, {\em SimpleStop} tends to choose the correct rank of $B$ better than {\em StrongStop}. This can be explained by the strict overestimation control property of {\em StrongStop} which might lead to underestimation in low signal case.
For strong signals, {\em StrongStop} is  better at avoiding overfitting.
\begin{figure}[tp]
\begin{center}
        \begin{tabular}{c}
                \includegraphics[scale=0.3]{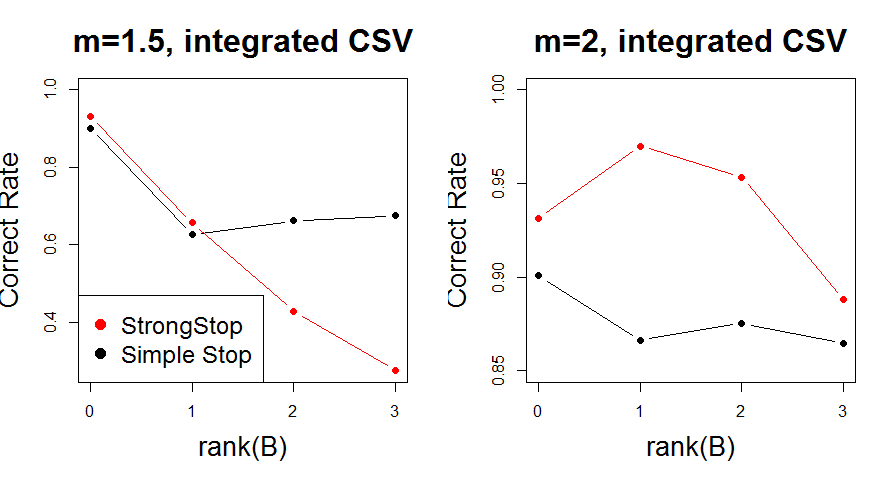} \\
        \end{tabular}
        \end{center}
        \caption{Rate of selecting the correct ${\rm rank}(B)$ versus ${\rm rank}(B)$ for the {\tt ICSV} procedure with $N=50$ and $p=10$. From the left to the right, $m$ increases from $1.5$ to $2$ where $m$ determines the magnitude of the signals in (\ref{simsetting}). The black dots and the red dots represent the results from {\em SimpleStop} and {\em StrongStop} respectively. For both stopping rules, $\alpha=0.1$ is used.\label{fig:corr}}
\end{figure}
\section{Estimating the Noise Level\label{sec:noise}}
       For the testing procedures {\tt CSV} and {\tt ICSV}, and confidence interval $CI_{k}(\mathbb{S})$, we have assumed that  the noise level $\sigma^2$ is known. In case the prior information of $\sigma^2$ is unavailable, the value of $\sigma^2$ needs to be estimated. In this section, we introduce a data-driven method for  estimating $\sigma^2$. 
        
        For the estimation of $\sigma^2$, it is popular to assume that the rank of $B$ is known. One of the simplest methods estimates $\sigma^2$ using mean of sum of squared residuals by
        \begin{eqnarray*}  
                \hat{\sigma}^2_{simple} = \frac{1}{N(p-\kappa)} \sum_{j=\kappa +1}^{p} d_{j}^2 
        \end{eqnarray*}
with known ${\rm rank}(B)=\kappa$.     
   
        Instead of using the rank of $B$, \cite{donoho2013optimal} use the median of the singular values of $Y$ as a robust estimator of $\sigma^2$ as follows:
        \begin{eqnarray}
                \label{medSigEst}
        \hat{\sigma}_{med}^2 = \frac{d_{med}^2}{N \cdot \mu_{\beta}}
        \end{eqnarray}
where $d_{med}$ is a median of  the singular values of $Y$ and $\mu_{\beta}$ is a median of a  Mar\v{c}enko-Pastur distribution with $\beta = N/p$. This estimator works under the assumption that ${\rm rank}(B) \ll \min(N, p)$.

        In this paper, we suggest three estimators, all of which make few assumptions.
        Our approach uses cross-validation and the sum of squared residuals as an extension of classical noise level estimator.
        For a fixed value of $\lambda$, we define our estimators $\hat{\sigma}^2_{\lambda}$, $\hat{\sigma}^2_{\lambda, df}$ and $\hat{\sigma}^2_{\lambda,df,c}$ as follows:
        \begin{eqnarray}
        \label{estSig}  
        \hat{\sigma}^2_{\lambda} &=& \frac{1}{N \cdot p} ||Y - \hat{B}_{\lambda}||^2_{F}  \nonumber\\
        \hat{\sigma}^2_{\lambda, df} &=& \frac{1}{N \cdot (p-df_{\hat{B}_{\lambda}})} ||Y - \hat{B}_{\lambda}||^2_{F}  \nonumber \\
        \hat{\sigma}^2_{\lambda, df, c} &=& \frac{1}{N \cdot (p-c \cdot df_{\hat{B}_{\lambda}})} ||Y - \hat{B}_{\lambda}||^2_{F} \mbox{ for } c \in [0,1]
        \end{eqnarray}
        with 
                \begin{eqnarray*}
                \hat{B}_{\lambda} &=& \underset{B \in \mathbb{R}^{N \times p}}{\mbox{argmin}} 
                \frac{1}{2} ||Y-B||_{F}^2 + \lambda ||B||_{*}  \\
        df_{\hat{B}_{\lambda}}  &=& \sum_{k=1}^{p} 1 \{ l_{\lambda, k} > 0\}
                \end{eqnarray*}                 
        where $l_{\lambda, k}$ denotes the $k^{th}$ singular value of $\hat{B}_{\lambda}$.
        The estimators $\hat{\sigma}^2_{\lambda}$ and $\hat{\sigma}^2_{\lambda, df}$ correspond to the ordinary mean  squared residual, with the latter accounting for the degrees of freedom \citep{reid2013study}.     
         With $c \in (0,1)$, the estimator $\hat{\sigma}^2_{\lambda, df, c}$ lies between $\hat{\sigma}^2_{\lambda}$ and $\hat{\sigma}^2_{\lambda, df}$. We use cross-validation for choosing the appropriate value of the regularization parameter $\lambda$.

        For these estimators, choosing an appropriate  value for the  regularization parameter $\lambda$ is important, since $\hat{B}_{\lambda}$ depends on $\lambda$. In penalized regression, it is common to use cross-validation for this purpose,  examining a grid of $\lambda$ values.
         Unlike the regression setting, however, here there is no outcome variable and thus it is not clear how to make predictions on left-out data points.
        
In this paper, we use {\em softImpute} algorithm \citep{mazumder2010spectral}. In the presence of missing values in a given data matrix, {\em softImpute}  carries out matrix completion with the following criterion:
        \begin{eqnarray}
                \label{softObj}
                \min_{B} \frac{1}{2}||P_{\Omega}(Y) - P_{\Omega}(B) ||_{F}^2 +\lambda||B||_{*}
        \end{eqnarray}
where $\Omega$ is an index set of observed data point with a function $P_{\Omega}(\cdot)$  such that $P_{\Omega}(Y)_{(i,j)} = Y_{i,j}$ if $(i, j) \in \Omega$ and $0$ otherwise. We define the prediction error of the unobserved values $\Omega$ as follows:
        \begin{eqnarray*}
                err_{\lambda}(\Omega) = || P_{\tilde{\Omega}}(Y) - P_{\tilde{\Omega}}(\hat{B}_{\lambda}^{S}(\Omega))||_{F}^{2}
        \end{eqnarray*} 
        where
            $\hat{B}_{\lambda}^{S}(\Omega)$ denotes the estimator of $B$ acquired from (\ref{softObj}) and  $\tilde{\Omega}$ denotes the index set of unobserved values ($\tilde{\Omega} = \Omega^c$).
        Using this prediction error, we carry out k-fold cross-validation, randomly generating $k$ non-overlapping leave-out sets of size $\frac{N \cdot p}{k}$ from $Y$. For a grid of $\lambda$ values, we compute the average of $err_{\lambda}$ for each $\lambda$ over the left-out data. We choose our $\lambda$ to be the minimizer of   the average $err_{\lambda}$ as in usual cross-validation (see e.g. \cite{hastie2009elements}). 
        \subsection{A study  of  noise level estimation}
         \label{subsec:sigSimulation}   
        We illustrate simulation examples of noise level estimation of the proposed methods $\hat{\sigma}^2_{\lambda}$, $\hat{\sigma}^2_{\lambda, df}$ and $\hat{\sigma}^2_{\lambda,df,c}$ in (\ref{estSig}) compared to $\hat{\sigma}^2_{med}$ in (\ref{medSigEst}). For the estimator $\hat{\sigma}^2_{\lambda,df,c}$, we use $c=\frac{2}{3}$ in ad-hoc. These approaches do not require predetermined knowledge of ${\rm rank}(B)$.

        Simulation settings are the same as in Section \ref{subsec:hyptest_sim} with $N=50$ and $p=10$. The true value of the noise level is $\sigma^2=1$ and for choosing $\lambda$, 20-fold cross-validation is used. Table \ref{tab:sigEst} illustrates the simulation results for the three proposed estimators.
\begin{table}[H]
        \caption{ 
        Simulation results for estimating the noise level $\sigma^2=1$ with $N=50$ and $p=10$. We vary ${\rm rank}(B)$ from $0$ to $3$, and $m$ from $0.5$ to $2.0$. Each column represents the mean estimated value of $\sigma^2$ (``{\bf Est}''), and standard error (``{\bf se}'') of the corresponding estimator.}
         
                \label{tab:sigEst} 
\begin{small}
        \hspace{-1.5cm}
        \begin{center}
        \begin{tabular}{c@{\extracolsep{0.5cm}}  c@{\extracolsep{0.2cm}} c@{\extracolsep{0.5cm}}  c@{\extracolsep{0.2cm}} c@{\extracolsep{0.5cm}}  c@{\extracolsep{0.2cm}} c@{\extracolsep{0.5cm}}  c@{\extracolsep{0.2cm}} c}
        \hline          
                \multirow{2}{*}{} & \multicolumn{2}{c}{$\hat{\sigma}^2_{\lambda_{CV}}$} & \multicolumn{2}{c}{$\hat{\sigma}^2_{\lambda_{CV}, df}$} & \multicolumn{2}{c}{$\hat{\sigma}^2_{\lambda_{CV},df,c}$} & \multicolumn{2}{c}{$\hat{\sigma}^2_{med}$} \\ \cline{2-3} \cline{4-5} \cline{6-7} \cline{8-9}
 \bf{m} & \bf Est & \bf se & \bf Est & \bf se & \bf Est & \bf se & \bf Est & \bf se \\ \hline
         \multicolumn{9}{c}{} \\[0.1pt]
         \multicolumn{9}{c}{${\rm rank}(B)=0$} \\
         0.0 & 0.863 & 0.230 & 0.990 & 0.230 & 0.926 & 0.210 & 0.996 & 0.084 \\
         \multicolumn{9}{c}{} \\[0.1pt]
         \multicolumn{9}{c}{${\rm rank} (B) = 1$} \\
         0.5 & 0.869 & 0.236 & 1.000 & 0.238 & 0.934 & 0.216 & 1.006 & 0.085 \\
         1.0 & 0.869 & 0.254 & 1.039 & 0.263 & 0.954 & 0.234 & 1.027 & 0.088 \\
         1.5 & 0.823 & 0.271 & 1.127 & 0.318 & 0.969 & 0.254 & 1.044 & 0.090 \\
         2.0 & 0.789 & 0.246 & 1.245 & 0.324 & 0.999 & 0.224 & 1.052 & 0.091 \\
         
         \multicolumn{9}{c}{} \\[0.1pt]
         \multicolumn{9}{c}{${\rm rank} (B) = 2$} \\
         0.5 & 0.871 & 0.260 & 1.061 & 0.280 & 0.962 & 0.237 & 1.039 & 0.089 \\
         1.0 & 0.784 & 0.262 & 1.321 & 0.355 & 1.025 & 0.244 & 1.093 & 0.096 \\
         1.5 & 0.700 & 0.288 & 1.611 & 0.431 & 1.052 & 0.201 & 1.121 & 0.100 \\
         2.0 & 0.646 & 0.211 & 1.762 & 0.471 & 1.047 & 0.196 & 1.134 & 0.102 \\
         
         \multicolumn{9}{c}{} \\[0.1pt]
         \multicolumn{9}{c}{${\rm rank} (B) = 3$} \\
         0.5 & 0.827 & 0.303 & 1.206 & 0.365 & 1.002 & 0.288 & 1.098 & 0.095 \\
         1.0 & 0.674 & 0.235 & 1.771 & 0.459 & 1.076 & 0.231 & 1.186 & 0.107 \\
         1.5 & 0.585 & 0.202 & 2.135 & 0.623 & 1.062 & 0.228 & 1.226 & 0.113 \\
         2.0 & 0.549 & 0.180 & 2.324 & 0.711 & 1.057 & 0.224 & 1.242 & 0.116 \\
         \hline
        \end{tabular} 
        \end{center} 
        \end{small}
\end{table}
        In this setting, $\hat{\sigma}^2_{\lambda_{CV}}$ decreases with larger ${\rm rank}(B)$ and signals while $\hat{\sigma}^2_{\lambda_{CV, df}}$ and $\hat{\sigma}^2_{med}$ increases. For large ${\rm rank}(B)$ and signals, $\hat{\sigma}^2_{\lambda_{CV, df, c}}$ shows good results, as compared to other methods.
        The poor performance of $\hat{\sigma}^2_{\lambda_{CV}, df}$ may be caused by the use of an improper  definition of $df$, the  degrees of freedom. 
        Following the definition of degrees of freedom by \cite{efron2004least}, our simulation result shows that the number of non-zero singular values does not coincide with degrees of freedom under our setting. Further investigation into the degrees of freedom is needed in future work.
        
        The competing method $\hat{\sigma}^2_{med}$ consistently shows a small standard deviation. However, with large ${\rm rank}(B)$, especially when ${\rm rank}(B) \geq p/2$, the procedures  over-estimates $\sigma^2$ due to the effect of the signals.
\section{Additional Examples\label{sec:sim}}
We discuss additional examples in this section. 
Section \ref{subsec:unknownSig} presents results of the proposed methods when the estimated noise level is used. 
In Section \ref{subsec:nonGaussian}, hypothesis testing results with non-Gaussian noise are illustrated. Section \ref{subsec:real_example} shows results on some real data.
\subsection{Simulation examples with unknown noise level\label{subsec:unknownSig}}
        In this section, we illustrate the results when estimated $\sigma^2$ value is used on simulated data. For the estimation of the noise level, we use $\hat{\sigma}^2_{\lambda_{CV}, df, c}$ and $\hat{\sigma}^2_{med}$ which showed good performance in Section \ref{subsec:sigSimulation}. As in Section \ref{subsec:sigSimulation}, for the estimator $\hat{\sigma}^2_{\lambda_{CV}, df, c}$, 20-fold cross-validation and $c=2/3$ is used. The simulation settings are the same as in Section \ref{subsec:hyptest_sim} with $N=50$ and $p=10$. We investigate the case of $m=1.5$.
\begin{figure}[p]
        \center
        \begin{tabular}{c}
                \includegraphics[height=1.5in, width=5in]{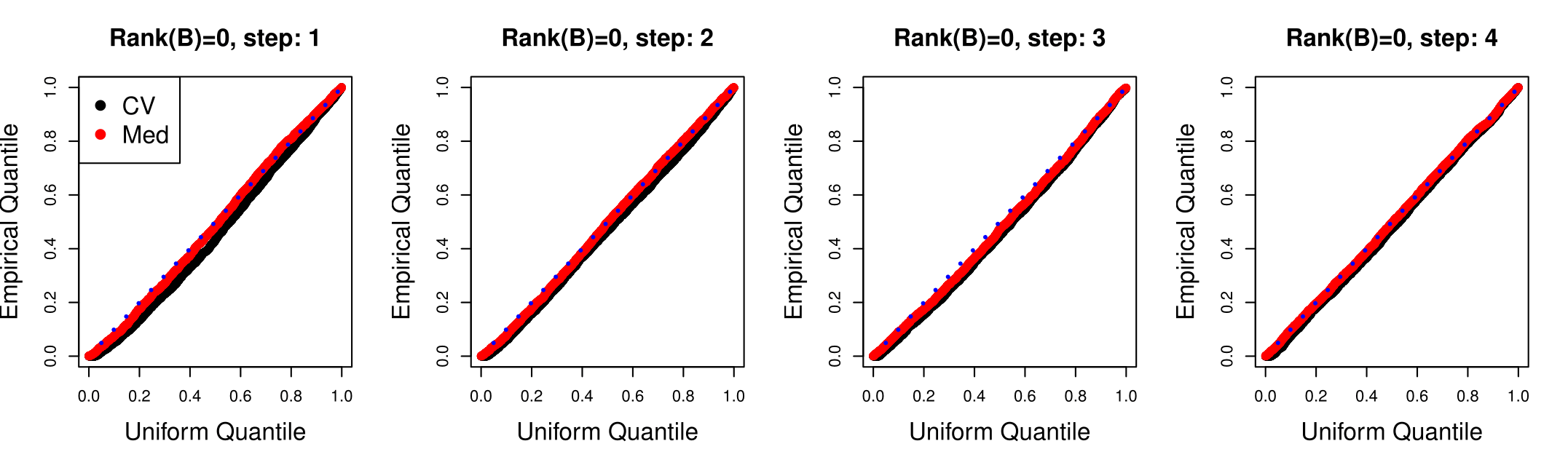} \\
                \includegraphics[height=1.5in, width=5in]{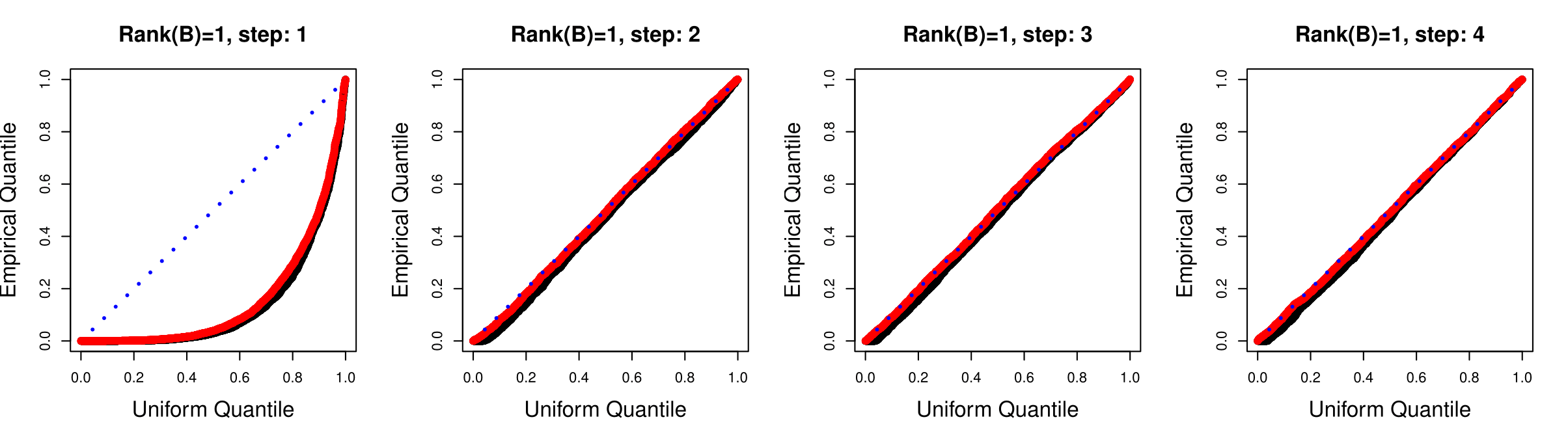} \\
                \includegraphics[height=1.5in, width=5in]{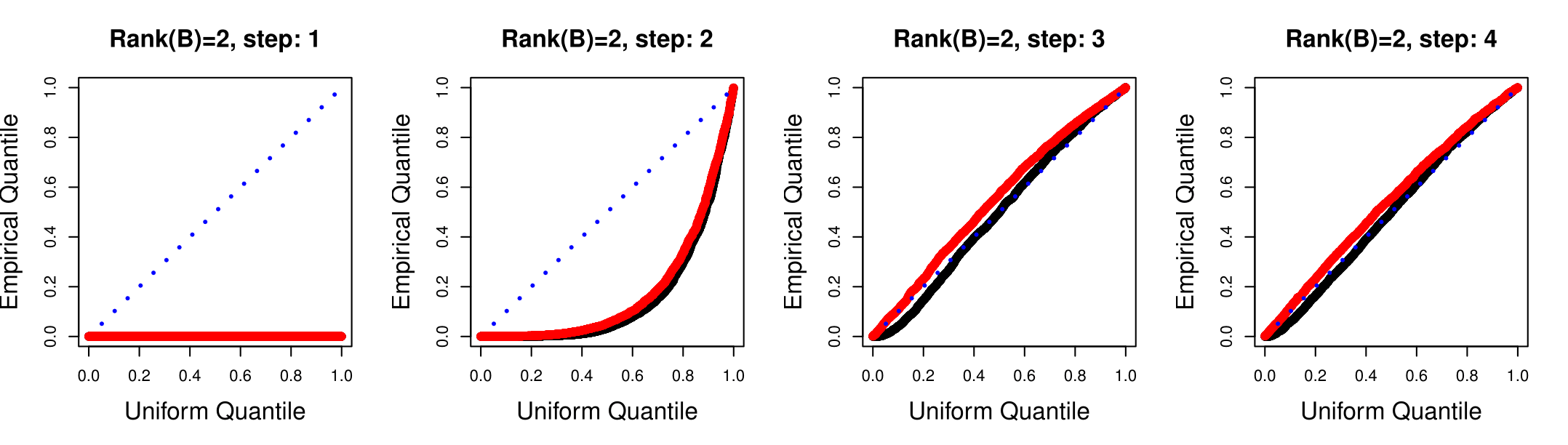} \\
                \includegraphics[height=1.5in, width=5in]{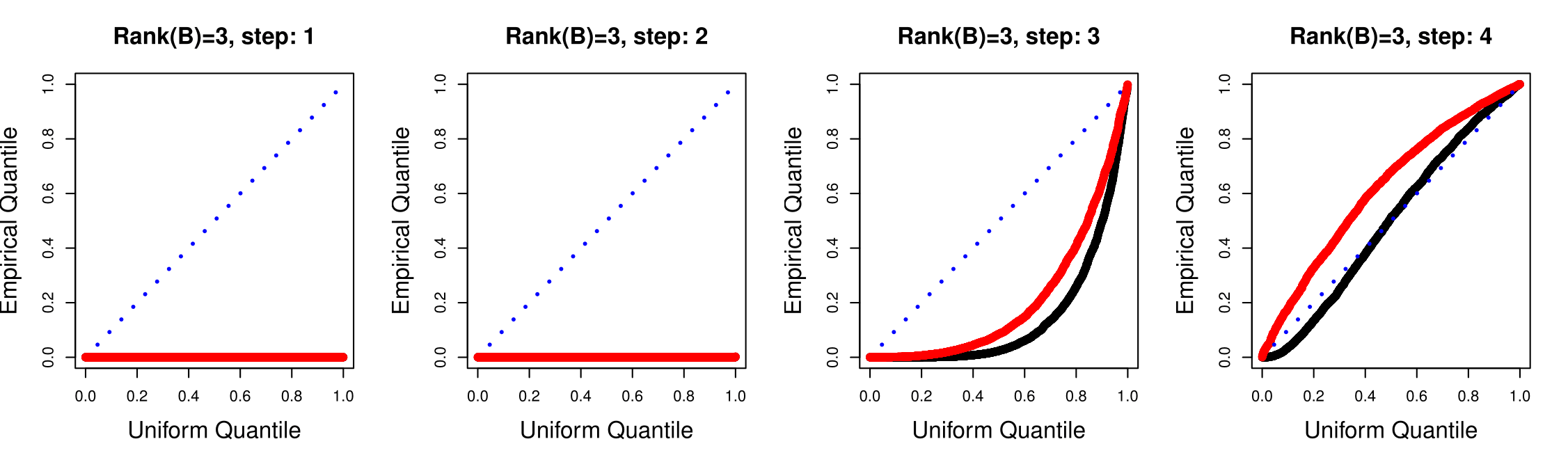} \\
        \end{tabular}
        \caption{ Quantile-quantile plots of the empirical quantiles of p-values with estimated $\sigma^2$ versus the  uniform quantiles at $m=1.5$ with $N=50$ and $p=10$. Emprical p-values are from the {\tt CSV} test. Each row represents ${\rm rank}(B)=0$ to ${\rm rank}(B)=3$ from the top to the bottom. Each column represents the $1^{st}$ test to the  $4^{th}$ test from the left to the right. The black dots (CV) and the red dots (med) represent p-values from using $\hat{\sigma}^2_{\lambda_{CV}, df, c}$ and $\hat{\sigma}^2_{med}$ respectively. For the estimator $\hat{\sigma}^2_{\lambda_{CV}}$, $c=2/3$ is used. \label{fig:pow_unknown}
        }
        \end{figure} 
\begin{table}[b]
        \caption{
        Simulation results for estimating the noise level $\sigma^2=1$ at $m=1.5$ with $N=50$ and $p=10$. We vary the rank of $B$ from $0$ to $3$. Shown are the mean estimated noise level (``{\bf Est}'') and standard error (``{\bf se}'') of the corresponding estimators.}
   \label{tab:sigEst_2} 
   \center
        \begin{tabular}{ c@{\extracolsep{0.5cm}}  c@{\extracolsep{0.2cm}} c@{\extracolsep{0.5cm}}  c@{\extracolsep{0.2cm}} c }
        \hline          
                \multirow{2}{*}{} & \multicolumn{2}{c}{$\hat{\sigma}^2_{\lambda_{CV}, df, c}$} & \multicolumn{2}{c}{$\hat{\sigma}^2_{med}$} \\ \cline{2-3} \cline{4-5} 
                ${\rm rank}(B)$ & \bf Est & \bf se & \bf Est & \bf se \\ \hline
                0 & 0.926 & 0.210 & 0.996 & 0.084 \\
                1 & 0.969 & 0.254 & 1.044 & 0.090 \\
                2 & 1.052 & 0.201 & 1.121 & 0.100 \\
                3 & 1.062 & 0.228 & 1.226 & 0.113 \\ 
                \hline
        \end{tabular} 
\end{table}

        Table \ref{tab:sigEst_2} illustrates the estimated $\sigma^2$ values we used for the testing procedure. Figure \ref{fig:pow_unknown} shows quantile-quantile plots of observed p-values obtained from using the estimated $\sigma^2$ versus the expected (uniform) quantiles. In quantile-quantile plots, both estimators of $\sigma^2$ show reasonable results in general, and for large ${\rm rank}(B)$, $\hat{\sigma}^2_{\lambda_{CV}, df, c}$ shows better result than $\hat{\sigma}^2_{med}$.  In terms of coverage rate of confidence interval, we can see from Figure \ref{fig:coverage_rate_unknown} that $\hat{\sigma}^2_{med}$ dominates for all cases, which might be due to small standard deviation of $\hat{\sigma}^2_{med}$ estimator. The estimation of ${\rm rank}(B)$ is presented in Table \ref{tab:estRank_unknown}. For the estimation,  {\em StrongStop} is applied to the {\tt CSV} p-values with level $\alpha=0.05$. The estimation performance seems to vary with the quality of the estimate of $\sigma^2$.

        \begin{figure}[tp] 
                \center
                \includegraphics[scale=.2]{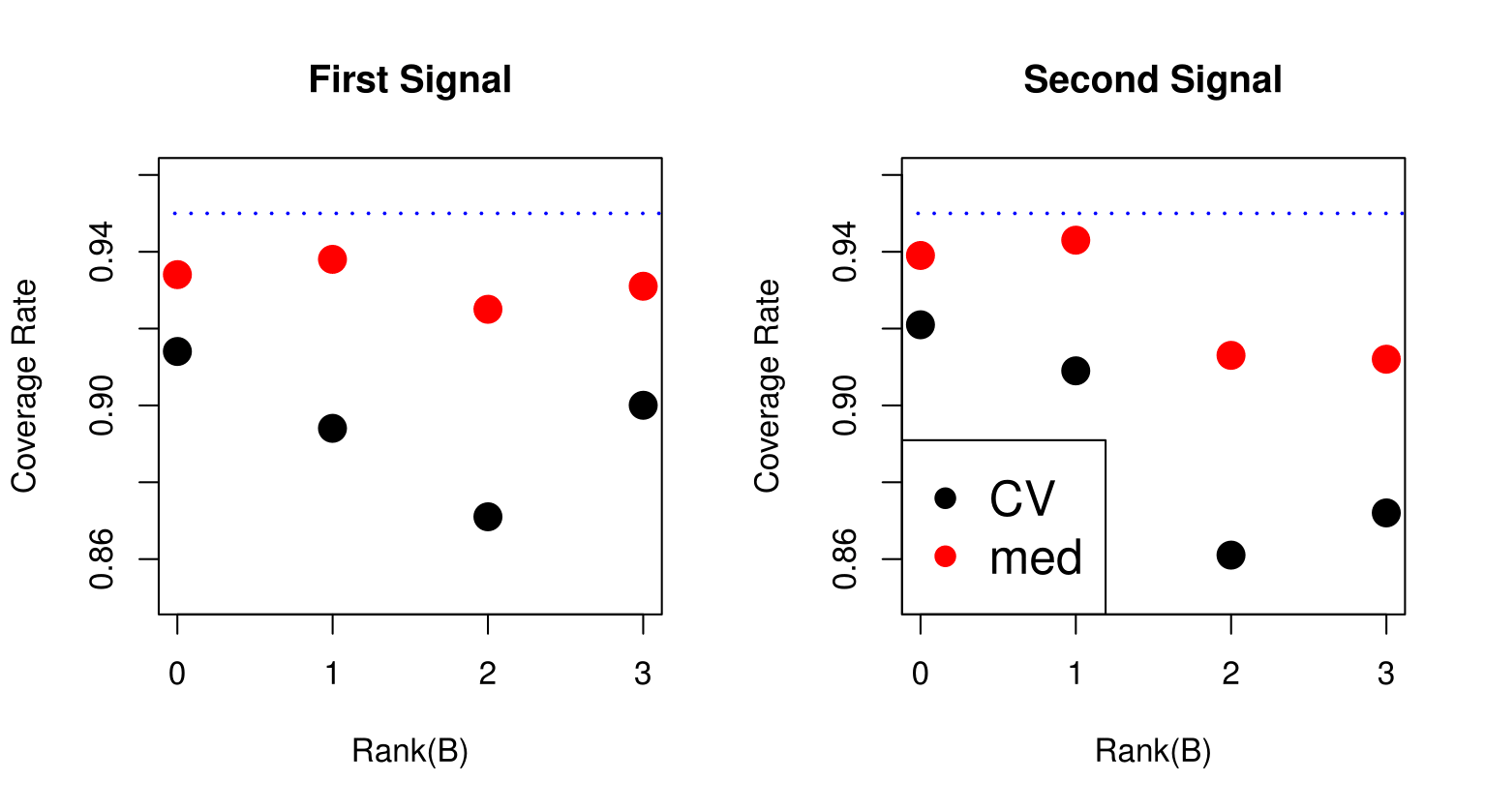} \\
        \caption{ Coverage rate versus ${\rm rank}(B)$ of the first two signal parameters at $m=1.5$ using estimated $\sigma^2$ with $N=50$ and $p=10$. Coverage rate denotes proportion of times  that  the constructed confidence interval from $CI_{k}(\mathbb{S})$ covered the true parameter. The black dots (CV) and the red dots (med) represent the estimation using $\hat{\sigma}^2_{\lambda_{CV}, df, c}$ and $ \hat{\sigma}^2_{med}$ respectively.\label{fig:coverage_rate_unknown}
        } 
\end{figure}
\begin{table}[h!]
 \caption{
        Simulation results for selecting an exact ${\rm rank}(B)$ from the {\tt CSV} test using estimated $\sigma^2$ at $m=1.5$ with $N=50$ and $p=10$. {\em StrongStop} is applied to sequential p-values with $\alpha=0.05$. We vary the rank of $B$ from 0 to 3. Shown are the rate of selecting the correct rank of $B$ (``{\bf Rate}'') and mean squared error (``{\bf MSE}'') using the corresponding estimator of $\sigma^2$.
}
        \label{tab:estRank_unknown}
        \center
        \begin{tabular}{c@{\extracolsep{0.5cm}}  c@{\extracolsep{0.2cm}} c@{\extracolsep{0.5cm}}    c@{\extracolsep{0.2cm}} c  }
        \hline
        \multirow{2}{*}{} & \multicolumn{2}{c}{$\hat{\sigma}^2_{\lambda_{CV}, df, c}$} & \multicolumn{2}{c}{$\hat{\sigma}^2_{med}$} \\ \cline{2-3} \cline{4-5}
        ${\rm rank}(B)$ & \bf Rate & \bf MSE & \bf Rate & \bf MSE \\ \hline
        0 & 0.894 & 3.704 & 0.948 & 0.063 \\
        1 & 0.455 & 4.243 & 0.486 & 0.514 \\
        2 & 0.231 & 1.080 & 0.157 & 0.853 \\
        3 & 0.181 & 0.845 & 0.026 & 0.975 \\ 
        \hline       
        \end{tabular} 
\end{table}
\subsection{Simulation example with non-Gaussian noise\label{subsec:nonGaussian}}
        Our testing procedure is based on an assumption of Gaussian noise. Here we investigate the performance of the {\tt CSV} test on simulated examples with the normality assumption of noise violated. Simulation settings are the same as in Section \ref{subsec:hyptest_sim} with $N=50$ and $p=10$ except for the noise distribution. We study the case of ${\rm rank}(B)=1$ with $m=1.5$ along with two sorts of noise distribution: heavy tailed and right skewed. The heavy tailed noise is drawn from $\sqrt{\frac{3}{5}} t_5$ where $t_5$ denotes  t-distribution with degrees of freedom=5, and the right skewed noises are drawn from $\sqrt{\frac{3}{10}} t_5 + \sqrt{\frac{1}{2}} ({\rm exp}(1)-1)$ where ${\rm exp}(1)$ denotes the exponential distribution with mean=1. In each case, noise entries are drawn i.i.d., and known value of $\sigma^2=1$ is used.
        \begin{figure}[t]
        \center
        \begin{tabular}{c}
                \includegraphics[height=1.5in, width=5in]{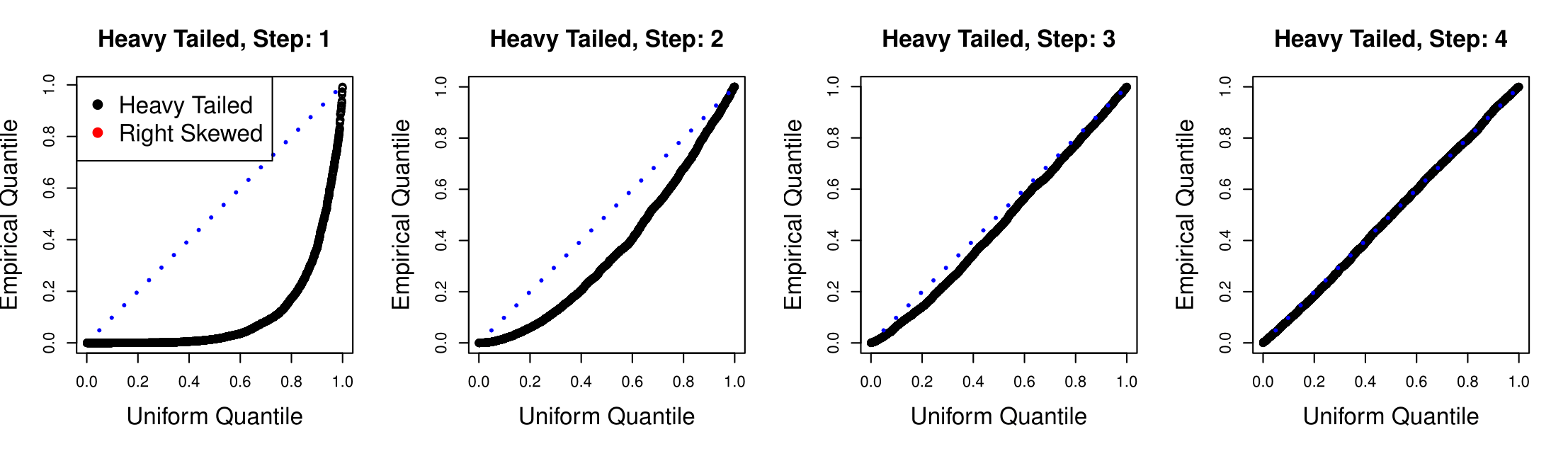} \\
                \includegraphics[height=1.5in, width=5in]{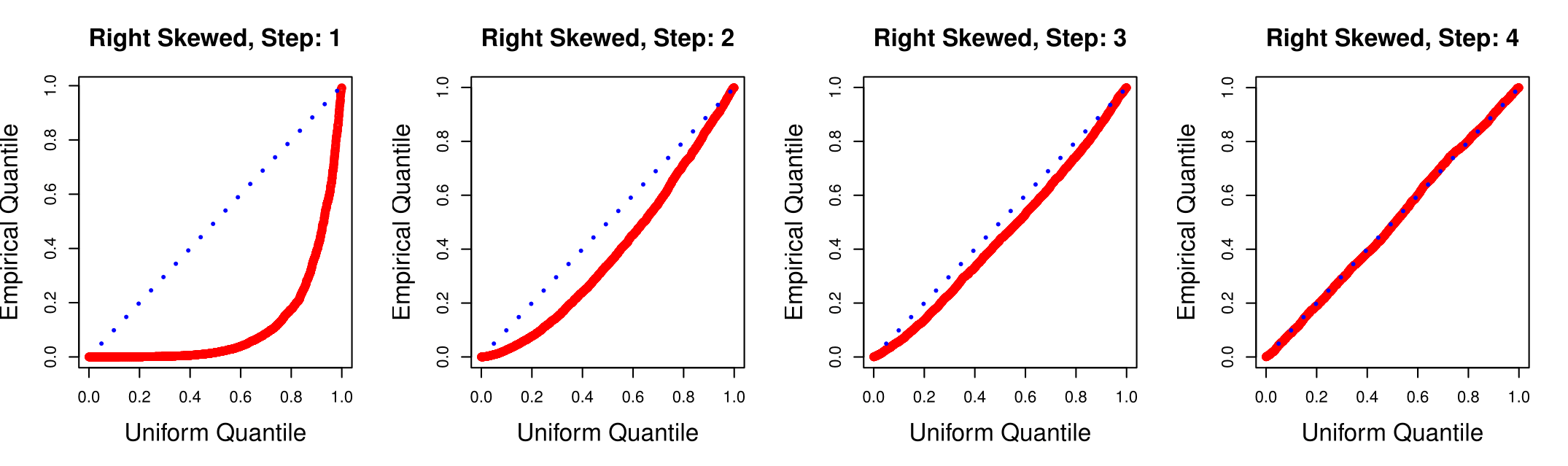} \\
        \end{tabular}
        \caption[]{ Quantile-quantile plots of empirical quantile of p-values versus uniform quantile when noises are drawn i.i.d. from the heavy tailed or the right skewed distribution with known value of $\sigma^2=1$ with $N=50$ and $p=10$. The true rank of $B$ is ${\rm rank(B)}=1$ with $m=1.5$. The top panels correspond to the heavy tailed noise from $\sqrt{\frac{3}{5}} t_5$ and the bottom panels correspond to the right skewed noise from $\sqrt{\frac{3}{10}} t_5 + \sqrt{\frac{1}{2}} ({\rm exp}(1)-1)$. Each column represents the $1^{st}$ test to the $4^{th}$ test from the left to  the right.
        }
        \label{fig:non_gaussian}
        \end{figure}        
        
        Figure \ref{fig:non_gaussian} shows quantile-quantile plots of the observed p-values versus expected (uniform) quantiles.  The top panels correspond to the heavy tailed noise from $\sqrt{\frac{3}{5}} t_5$ and the bottom panels correspond to the right skewed noise from $\sqrt{\frac{3}{10}}t_5 + \sqrt{\frac{1}{2}}({\rm exp}(1)-1)$.
        Quantile-quantile plots from both types of noise show that the p-values  deviate slightly from a uniform distribution under the null $H_{0,2}: {\rm rank}(B) \leq 1$. For further steps, p-values lie closer to the reference line.  
        
        The nonconformity shown in early steps under the null hypothesis is not surprising considering the construction of the {\tt CSV} procedure based on Gaussian noise. In future work, we will investigate whether the procedures 
introduced here can be extended to a method robust to non-normality using data-oriented method such as bootstrap.

\subsection{Real data example\label{subsec:real_example}}
        In this section, we revisit the real data example mentioned in Figure \ref{fig:scor}.
        We apply the {\tt CSV} test to the data  of  examination marks of 88 students on 5 different topics of Mechanics, Vectors, Algebra, Analysis and Statistics \citep[p. 3-4]{mardia1979multivariate}, and determine the number of principal components to retain for PCA. 
        
        In this data, Mechanics and Vectors were closed book exams while the other topics were open book exam.
        We use $\hat{\sigma}^2_{\lambda_{CV}, df, c}=75.957$ and $\hat{\sigma}^2_{med}=131.332$ for the estimated noise level.
         For $\hat{\sigma}^2_{\lambda_{CV}, df, c}$, 20-fold cross-validation is used with $c=2/3$. The {\tt CSV} test results are presented in Table \ref{tab:real_pval}. The estimated ${\rm rank}(B)$ is 2 with $\hat{\sigma}^2_{\lambda_{CV}, df, c}$ estimator and 1 with $\hat{\sigma}^2 _{med}$ with level $\alpha=0.05$ using {\em StrongStop}. Thus, in PCA we may use one or two principal components depending on our choice of the noise level. In this example, one or two principal components makes sense as these 5 topics cover closely related areas.  

\begin{table}[h!]
        \caption{
        P-values at each step (``{\bf Step}'') and the selected number of principal components (``{\bf Selected}'') from the {\tt CSV} test using the estimated noise level on the examination marks of 88 students on five different topics (Mechanics, Vectors, Algebra, Analysis and Statistics). {\em StrongStop} is applied to select the number of principal components at level $\alpha=0.05$. The estimated value of $\sigma^2$ is 75.957 for the estimator $\hat{\sigma}^2_{\lambda_{CV}, df, c}$ and 131.332 for the estimator $\hat{\sigma}^2_{med}$.
        \label{tab:real_pval}} 
        \center
        \begin{tabular}{ c@{\extracolsep{0.5cm}} c@{\extracolsep{0.2cm}} c }
        \hline
        \bf{Step} & $\hat{\sigma}^2_{\lambda_{CV}, df, c}$ & $\hat{\sigma}^2_{med}$  \\ \hline      
        1 & 0.000 & 0.000 \\
        2 & 0.000 & 0.015 \\
        3 & 0.001 & 0.573 \\
        4 & 0.093 & 0.940 \\
        \hline
        \hline
            & $\hat{\sigma}^2_{\lambda_{CV}, df, c}$ & $\hat{\sigma}^2_{med}$  \\ \hline
        \bf Selected  & 2 & 1 \\
        \hline  
        \end{tabular} 
\end{table}
\section{Conclusions\label{sec:conclusion}}
        In this paper, we have proposed distribution-based methods for choosing the number of principal components of a data matrix. We have suggested novel methods  both  for hypothesis testing  and the construction of confidence intervals of the signals. The methods have exact type I error control and show promising results in simulated examples. We have also introduced data-based methods for estimating the noise level.
        
        There are many topics that deserve further investigation. In following studies, the analysis of power of the suggested tests and the width of the constructed confidence interval will be investigated. 
        Also, application of the methods to high dimensional data using numerical approximations will be explored.
        For multiple hypothesis testing corrections to be properly applied, we will study the dependence structure of the p-values in different steps.
       In addition, for robustness to non-Gaussian noise, bootstrap versions of this procedure will be investigated.
        Future work may involve a notion of degrees of freedom of the  spectral estimator of the signal matrix.  These extensions may  lead to improvement in noise level estimation.
        
        Variations of these procedures can potentially be applied to canonical correlation analysis (CCA) and linear discriminant analysis (LDA), and these are topics for future work.

\appendix

\section*{Appendix}

\subsection{Lemma \ref{lem:general_result} and Proof}
\begin{lemma} 
        \label{lem:general_result} 
        For $Y \in \mathbb{R}^{N \times p}$ where $Y \sim N(B, \sigma^2 I_{N} \otimes I_{p})$, we write the singular value decomposition of $Y$ by $Y = U_Y D_Y V_Y^{t}$ where $D_Y = {\rm diag}(y_1, \cdots, y_p)$ with $y_1 \geq \cdots \geq y_p \geq 0$. Without loss of generality we assume $N \geq p$. As in Section \ref{subsec:csv}, $U_Y^{r}, U_Y^{-r}, V_Y^{r}$ and $V_Y^{-r}$ denote submatrices of $U_Y$ and $V_Y$.
Writing the law by $\mathcal{L}(\cdot)$, the density of $\mathcal{L}(\cdot)$ by $d \mathcal{L}$, ${\rm diag}(y_1, \cdots, y_r)$ by $D_Y^{r}$, and ${\rm diag}(y_{r+1}, \cdots, y_p)$ by $D_Y^{-r}$, we have
        \begin{eqnarray*}
        \frac{ d\mathcal{L}( D_Y^{-r}, U_Y^{-r}, V_Y^{-r}  | U_Y^{r}, V_Y^{r}, D_Y^{r}) }{d\mathcal{L}(D_Y^{-r}, U_Y^{-r}, V_Y^{-r} | U_Y^{r}, V_Y^{r}, D_Y^{r}, B=0 )} \propto e^{ {\rm tr}(U_Y^{-r} D_Y^{-r} V_Y^{-r^T})^{T}B}.
        \end{eqnarray*}
\end{lemma}
\begin{proof}
Without loss of generality, we assume $\sigma^2 = 1$.
Writing the density function of $Y$ as $p_{Y,B}(\cdot)$, when $B=0$, we have
\begin{eqnarray}
        \label{lem:null_case}
&&d\mathcal{L}(D_Y, U_Y, V_Y |B=0)  =  p_{Y,0}(U_Y D_Y V_Y^{T})|J(D_Y, U_Y, V_Y)| \nonumber \\
&& \mbox{\hspace{0cm} } \propto   e^{-\frac{1}{2} {\rm tr}((U_Y D_Y V_Y^{T})^{T}(U_Y D_Y V_Y^{T}))} \prod_{i=1}^{p} y_{i}^{N-p}\prod_{i<j} (y_i^2 - y_j^2) d\mu_{N, p}(U_Y) d \mu_{p,p}(V_Y) \nonumber \\ 
&& \mbox{\hspace{0.5cm} }\propto   F_{N,p}(y_1, \cdots, y_p) d\mu_{N, p}(U_Y) d \mu_{p,p}(V_Y) 
\end{eqnarray}
where 
\begin{eqnarray} 
        \label{lem:singular_dist}
        F_{N,p}(y_1, \cdots, y_p) = C_{N,p} e^{-\frac{\sum_{i=1}^{p}y_{i}^2}{2}} \prod_{i=1}^{p} y_i^{N-p} \prod_{i<j}(y_i^2 -y_j^2)
\end{eqnarray}
denotes the density of singular values of a $N \times p$ Gaussian matrix from $N(0, I_{N} \otimes I_{p})$ with a normalizing constant $C_{N,p}$, $J(\cdot)$ is the Jacobian of mapping $Y \rightarrow (D_Y, U_Y, V_Y)$, $\mu_{N,p}$ denotes Haar probability measure on $O(N)$ under the map $U \rightarrow U^p$, and $\mu_{p,p}$ denotes Haar probability measure on $O(p)$. Here, $O(m)$ denotes an orthogonal group of $m \times m$ matrices. 

        Note that, $d \mathcal{L}(D_Y, U_Y, V_Y | B=0) \propto  F_{N,p} d\mu_{N, p}(U_Y) d \mu_{p,p}(V_Y) $ can be alternatively derived from $D_Y \perp (U_Y, V_Y)$, and $U_Y \perp V_Y$. $U_Y$ and $V_Y$ have Haar probability measure on $O(n)$ under mapping $U \rightarrow U^{p}$ and on $O(p)$ respectively, since the distribution of $U_Y$ and $V_Y$ are rotation invariant. The determinant of the Jacobian $|J(D_Y, U_Y, V_Y)|$ can be calculated either explicitly or from this relation.
        
        As in the case of $B=0$, for general $B$ we have,
\begin{eqnarray}
\label{app:genlaw}
&& d \mathcal{L}(D_Y, U_Y, V_Y ) = p_{Y,B}(U_Y D_Y V_Y^{T})|J(D_Y, U_Y, V_Y)| \nonumber \\
&& \hspace{0.5cm} \propto  e^{ -\frac{1}{2}{\rm tr}( (U_Y D_Y V_Y^{T}-B)^{T}(U_Y D_Y V_Y^{T}-B) ) }|J(D_Y, U_Y, V_Y)| \nonumber \\
&& \hspace{0.5cm} \propto   e^{ {\rm tr}(U_Y D_Y V_Y^{T})^{T}B} \cdot e^{-\frac{1}{2}{\rm tr}B^T B} p_{Y,0}(U_Y D_Y V_Y^{T})|J(D_Y, U_Y, V_Y)| \nonumber \\
&& \hspace{0.5cm} \propto  e^{ {\rm tr}(U_Y D_Y V_Y^{T})^{T}B}  \cdot F_{N,p}(y_1, \cdots, y_p) d\mu_{N, p}(U_Y) d \mu_{p,p}(V_Y) \mbox{ from (\ref{lem:null_case})}.
\end{eqnarray}
Therefore, with given $U_Y^r$, $V_Y^{r}$, and $D_Y^{r}$, we have
\begin{eqnarray*}
&& d \mathcal{L}(D_Y, U_Y, V_Y | U_Y^r, V_Y^r, D_Y^r) \\
&& \hspace{0.5cm} \propto   e^{ {\rm tr}(U_Y D_Y V_Y^{T})^{T}B}  \cdot F_{N,p}(y_1, \cdots, y_p) d\mu_{N, p}(U_Y) d \mu_{p,p}(V_Y)  \\
&& \hspace{0.5cm} = e^{ {\rm tr} (U_Y^{r} D_Y^{r} V_Y^{r^T})^{T}B } \cdot e^{ {\rm tr} (U_Y^{-r} D_Y^{-r} V_Y^{-r^T})^{T}B } \cdot F_{N,p}(y_1, \cdots, y_p) d\mu_{N, p}(U_Y) d \mu_{p,p}(V_Y) \\
&& \hspace{0.5cm} \propto e^{ {\rm tr} (U_Y^{-r} D_Y^{-r} V_Y^{-r^T})^{T}B } F_{N,p}(y_1, \cdots, y_p) d\mu_{N, p}(U_Y) d \mu_{p,p}(V_Y)
\end{eqnarray*}
since $Y=U_Y D_Y V_Y^{T}=U_Y^{r}D_Y^{r}V_Y^{r} + U_Y^{-r}D_Y^{-r}V_Y^{-r}$, and thus
\begin{eqnarray}
\label{lem:condLaw}
&&d \mathcal{L}(D_Y, U_Y, V_Y | U_Y^r, V_Y^r, D_Y^r)  \propto \nonumber\\
&& \hspace{2cm} e^{ {\rm tr} (U_Y^{-r} D_Y^{-r} V_Y^{-r^T})^{T}B } F_{N,p}(y_1, \cdots, y_p) d\mu_{N, p}(U_Y) d \mu_{p,p}(V_Y)
\end{eqnarray}
and equivalently,
\begin{eqnarray*}
\frac{d\mathcal{L}(D_Y^{-r}, U_Y^{-r}, V_Y^{-r} | U^{r}, V^{r}, D^{r})}{d\mathcal{L}(D_Y^{-r}, U_Y^{-r}, V_Y^{-r} | U^{r}, V^{r}, D^{r}, B=0)} \propto e^{ {\rm tr}(U_Y^{-r} D_Y^{-r} V_Y^{-r^T})^{T}B} .  
\end{eqnarray*}
\end{proof}

\subsection{Proof of Theorem \ref{lem:csv}}
\begin{proof}
We follow the notations in Lemma \ref{lem:general_result}. Without loss of generality we assume $\sigma^2 =1$. When $U_{Y}^{-(k-1)}U_{Y}^{-(k-1) ^T} B V_{Y}^{-(k-1)}V_{Y}^{-(k-1)^T} = 0_{N \times p}$, then $e^{ {\rm tr}(U_Y^{-(k-1)}D_Y^{-(k-1)} V_Y^{-(k-1)^T})^{T}B}=1$. Thus, from Lemma \ref{lem:general_result} and (\ref{lem:condLaw}), we have
\begin{eqnarray*} 
        \label{lem:propEq} 
&& d \mathcal{L}(D_Y^{-(k-1)}, U_Y^{-(k-1)}, V_Y^{-(k-1)}| D_Y^{(k-1)}, U_Y^{(k-1)}, V_Y^{(k-1)}) \\ 
&& \hspace{1cm}\propto  d \mathcal{L}(D_Y^{-(k-1)}, U_Y^{-(k-1)}, V_Y^{-(k-1)}| D_Y^{(k-1)}, U_Y^{(k-1)}, V_Y^{(k-1)}, B=0) \\
&& \hspace{1cm} \propto   F_{N,p}(y_1, \cdots, y_p) d\mu_{N, p}(U_Y) d \mu_{p,p}(V_Y)
\end{eqnarray*}
where $F_{N,p}$ denotes the density of singular values of a $N \times p$ Gaussian matrix from $N(0, I_{N} \otimes I_{p} )$ as in (\ref{lem:singular_dist}). 
Therefore, we have
\begin{eqnarray} 
        \label{thm:csv_integrand}
&& d \mathcal{L}(y_k |y_i, i \neq k, U_Y^{(k-1)}, V_Y^{(k-1)}) \nonumber \\
&& \hspace{1cm} \propto  \prod_{i \neq k} |y_i^2 - y_k^2| e^{-\frac{y_k^2}{2}} y_k^{N-p} I_{\{ y_k \in (y_{k+1}, y_{k-1})\}} d\mu_{N, p}(U_Y) d \mu_{p,p}(V_Y).
\end{eqnarray}

        Note that (\ref{thm:csv_integrand}) is the integrand of the {\tt CSV} test statistic $\mathbb{S}_{k,0}$ with some cancellation from the fraction, and thus $\mathbb{S}_{k,0}$ is the probability of having the $k^{th}$ singular value that is bigger than the observed one given all the other singular values, $U_Y^{(k-1)}$ and $V_Y^{(k-1)}$.
        
        Thus, if the observation $Y$ is actually generated under\\ $U_{Y}^{-(k-1)}U_{Y}^{-(k-1) ^T} B V_{Y}^{-(k-1)}V_{Y}^{-(k-1)^T} = 0_{N \times p}$, then the density of the $k^{th}$ singular value becomes (\ref{thm:csv_integrand}) up to a constant. Consequently, when writing the observed singular values by $d_1, \cdots, d_p$, and the observed singular value decomposition of $Y$ by $U_Y D_Y V_Y^{T}$ without confusion, 
\begin{eqnarray*}
\mathbb{S}_{k,0} &=& \int_{d_k}^{d_{k-1}} d \mathcal{L}(y_k |d_i, i \neq k, U_Y^{(k-1)}, V_Y^{(k-1)}) dy_k \\
&=& P(y_k \geq d_k |  d_i, i \neq k , U_Y^{(k-1)}, V_Y^{(k-1)}) \\
                                  &\sim & {\rm Unif}(0,1).
\end{eqnarray*}
The denominator of $\mathbb{S}_{k,0}$ works as a normalizing constant.
\end{proof}

\subsection{Proof of Theorem~\ref{lem:integrated_csv}}
\begin{proof}
We follow the notation of Lemma \ref{lem:general_result} and without loss of generality, assume $\sigma^2 = 1$. Under $U^{-(k-1)}U^{-(k-1) ^T} B V^{-(k-1)}V^{-(k-1)^T} = 0_{N \times p}$, we have $e^{ {\rm tr}(U_Y^{-(k-1)}D_Y^{-(k-1)} V_Y^{-(k-1)^T})^{T}B}=1$. Then, as in Theorem \ref{lem:csv}, we have
\begin{eqnarray*}
        && d\mathcal{L}(D_Y^{-(k-1)}, U_Y^{-(k-1)}, V_Y^{-(k-1)} | D_Y^{(k-1)}, U_Y^{(k-1)}, V_Y^{(k-1)}) \\
& & \hspace{1cm }\propto    F_{N, p}(y_1, \cdots, y_p)        d\mu_{N, p}(U_Y) d \mu_{p,p}(V_Y)
\end{eqnarray*}
and therefore,
\begin{eqnarray}
        \label{lem:integrated_csv_integrand}
&& d\mathcal{L}(y_k, U_Y^{-(k-1)}, V_Y^{-(k-1)} | D_Y^{(k-1)}, U_Y^{(k-1)}, V_Y^{(k-1)}) \nonumber \\
&& \hspace{1cm} \propto   \left( \int \cdots \int F_{N, p}(y_1, \cdots, y_p) d y_p \cdots d y_{k+1} \right) \cdot d\mu_{N, p}(U_Y) d \mu_{p,p}(V_Y) \nonumber \\ 
&& \hspace{1cm} \propto  g(y_k) d\mu_{N, p}(U_Y) d \mu_{p,p}(V_Y)
\end{eqnarray}
where $g(\cdot)$ is defined in (\ref{gfun}). Note that (\ref{lem:integrated_csv_integrand}) is the integrand of the {\tt ICSV} test statistic $\mathbb{V}_{k,0}$, and $\mathbb{V}_{k,0}$ is the conditional survival function of the $k^{th}$ singular value given $D_Y^{(k-1)}, U_Y^{(k-1)}$ and $V_Y^{(k-1)}$.

        Thus, if the observation is under $U_{Y}^{-(k-1)}U_{Y}^{-(k-1) ^T} B V_{Y}^{-(k-1)}V_{Y}^{-(k-1)^T} = 0_{N \times p}$, then its $k^{th}$ singular value has the conditional     density of (\ref{lem:integrated_csv_integrand}) upto a constant. Consequently, when writing the observed $k^{th}$ singular value by $d_k$, and the observed singular value decomposition of $Y$ by $U_Y D_Y V_Y^{T}$ without confusion, 
\begin{eqnarray*}
\mathbb{V}_{k,0} &=& \int_{d_k}^{d_{k-1}}  d\mathcal{L}(y_k, U_Y^{-(k-1)}, V_Y^{-(k-1)} | D_Y^{(k-1)}, U_Y^{(k-1)}, V_Y^{(k-1)}) dy_k \\
&=& P(y_k \geq {d}_k | D_Y^{(k-1)}, U_Y^{(k-1)}, V_Y^{(k-1)}) \\
&\sim & {\rm Unif}(0,1).
\end{eqnarray*} 
The denominator of $\mathbb{V}_{k,0}$ works as a normalizing constant.
\end{proof}

\subsection{Proof of Theorem \ref{exactCI}}
\begin{proof}
We follow the notations in Lemma \ref{lem:general_result} and without loss of generality, assume $\sigma^2=1$. From (\ref{app:genlaw}), we have
\begin{eqnarray*}
&& d \mathcal{L}(D_Y, U_Y, V_Y) \\
&& \hspace{1cm}\propto e^{\rm{tr}(U_Y D_Y V_Y^T)^T B} \cdot F_{N,p}(y_1, \cdots, y_p) d \mu_{ N, p}(U_Y) d \mu_{p,p}(V_Y) \\
&& \hspace{1cm} = e^{\sum_{k=1}^p y_k \tilde{\Lambda}_k} \cdot F_{N,p}(y_1, \cdots, y_p) d \mu_{N, p}(U_Y) d \mu_{p,p}(V_Y) 
\end{eqnarray*}
where $\tilde{\Lambda}_k = \langle U_{Y,k}V_{Y,k}^{T}, B\rangle$ as defined in (\ref{tildeLambda_k}).
Therefore, we have 
\begin{eqnarray}
\label{app:ci_integrand}
&&d\mathcal{L}(y_k, | Y_i, i \neq k, U_Y, V_Y) \nonumber \\
&&\hspace{1cm} \propto e^{y_k\tilde{\Lambda}_k}\cdot F_{N,p}(y_1, \cdots, y_p) d \mu_{ N, p}(U_Y) d \mu_{ p,p}(V_Y) \nonumber\\
&&\hspace{1cm} \propto e^{-\frac{1}{2}(y_k - \tilde{\Lambda}_k)^2} y_k^{N-p} \prod_{j\neq k}^p |y_k^2 - y_j^2| 1 \{ 0 \leq y_p \leq \cdots \leq y_1\}.
\end{eqnarray}
As (\ref{app:ci_integrand}) is the integrand of $\mathbb{S}_{k, \tilde{\Lambda}_k}$, when writing the observed singular values by $d_1, \cdots, d_p$, and the observed singular value decomposition of $Y$ by $U_YD_YV_Y^T$ without confusion, we have
\begin{eqnarray*}
\mathbb{S}_{k, \tilde{\Lambda}_k} &=& \int_{d_k}^{d_{k-1}} d\mathcal{L}(y_k, | d_i, i \neq k, U_Y, V_Y) dy_k \\
& \sim & {\rm Unif}(0,1).
\end{eqnarray*}
Here, the denominator of $S_{k, \tilde{\Lambda}_k}$ works as a normalizing constant.
\end{proof}
\section*{Acknowledgements}
We would like to thank Boaz Nadler and Iain Johnstone for helpful conversations.
Robert Tibshirani was supported by NSF grant DMS-9971405 and NIH grant N01-HV-28183.
\bibliographystyle{imsart-nameyear}
\bibliography{bib2}

\begin{thebibliography}{17}

\bibitem[\protect\citeauthoryear{Cai, Cand{\`e}s and
  Shen}{2010}]{cai2010singular}
\begin{barticle}[author]
\bauthor{\bsnm{Cai},~\bfnm{Jian-Feng}\binits{J.-F.}},
  \bauthor{\bsnm{Cand{\`e}s},~\bfnm{Emmanuel~J.}\binits{E.~J.}} \AND
  \bauthor{\bsnm{Shen},~\bfnm{Zuowei}\binits{Z.}}
(\byear{2010}).
\btitle{A singular value thresholding algorithm for matrix completion}.
\bjournal{SIAM J. Optim.}
\bvolume{20}
\bpages{1956--1982}.
\bdoi{10.1137/080738970}
\bmrnumber{2600248 (2011c:90065)}
\end{barticle}
\endbibitem

\bibitem[\protect\citeauthoryear{Efron et~al.}{2004}]{efron2004least}
\begin{barticle}[author]
\bauthor{\bsnm{Efron},~\bfnm{Bradley}\binits{B.}},
  \bauthor{\bsnm{Hastie},~\bfnm{Trevor}\binits{T.}},
  \bauthor{\bsnm{Johnstone},~\bfnm{Iain}\binits{I.}} \AND
  \bauthor{\bsnm{Tibshirani},~\bfnm{Robert}\binits{R.}}
(\byear{2004}).
\btitle{Least angle regression}.
\bjournal{Ann. Statist.}
\bvolume{32}
\bpages{407--499}.
\bnote{With discussion, and a rejoinder by the authors}.
\bdoi{10.1214/009053604000000067}
\bmrnumber{2060166 (2005d:62116)}
\end{barticle}
\endbibitem

\bibitem[\protect\citeauthoryear{Gavish and Donoho}{2014}]{donoho2013optimal}
\begin{barticle}[author]
\bauthor{\bsnm{Gavish},~\bfnm{Matan}\binits{M.}} \AND
  \bauthor{\bsnm{Donoho},~\bfnm{David~L.}\binits{D.~L.}}
(\byear{2014}).
\btitle{The optimal hard threshold for singular values is {$4/\sqrt 3$}}.
\bjournal{IEEE Trans. Inform. Theory}
\bvolume{60}
\bpages{5040--5053}.
\bdoi{10.1109/TIT.2014.2323359}
\bmrnumber{3245370}
\end{barticle}
\endbibitem

\bibitem[\protect\citeauthoryear{G'Sell et~al.}{2013}]{g2013false}
\begin{barticle}[author]
\bauthor{\bsnm{G'Sell},~\bfnm{Max~Grazier}\binits{M.~G.}},
  \bauthor{\bsnm{Wager},~\bfnm{Stefan}\binits{S.}},
  \bauthor{\bsnm{Chouldechova},~\bfnm{Alexandra}\binits{A.}} \AND
  \bauthor{\bsnm{Tibshirani},~\bfnm{Robert}\binits{R.}}
(\byear{2013}).
\btitle{Sequential Selection Procedures And False Discovery Rate Control}.
\bnote{Preprint. Available at
  \href{http://arxiv.org/abs/1309.5352}{arXiv:1309.5352}.}
\end{barticle}
\endbibitem

\bibitem[\protect\citeauthoryear{Hastie, Tibshirani and
  Friedman}{2009}]{hastie2009elements}
\begin{bbook}[author]
\bauthor{\bsnm{Hastie},~\bfnm{Trevor}\binits{T.}},
  \bauthor{\bsnm{Tibshirani},~\bfnm{Robert}\binits{R.}} \AND
  \bauthor{\bsnm{Friedman},~\bfnm{Jerome}\binits{J.}}
(\byear{2009}).
\btitle{The elements of statistical learning},
\bedition{second} ed.
\bseries{Springer Series in Statistics}.
\bpublisher{Springer, New York}
\bnote{Data mining, inference, and prediction}.
\bdoi{10.1007/978-0-387-84858-7}
\bmrnumber{2722294 (2012d:62081)}
\end{bbook}
\endbibitem

\bibitem[\protect\citeauthoryear{James}{1964}]{james1964distributions}
\begin{barticle}[author]
\bauthor{\bsnm{James},~\bfnm{Alan~T.}\binits{A.~T.}}
(\byear{1964}).
\btitle{Distributions of matrix variates and latent roots derived from normal
  samples}.
\bjournal{Ann. Math. Statist.}
\bvolume{35}
\bpages{475--501}.
\bmrnumber{0181057 (31 \#5286)}
\end{barticle}
\endbibitem

\bibitem[\protect\citeauthoryear{Johnstone}{2001}]{johnstone2001distribution}
\begin{barticle}[author]
\bauthor{\bsnm{Johnstone},~\bfnm{Iain~M.}\binits{I.~M.}}
(\byear{2001}).
\btitle{On the distribution of the largest eigenvalue in principal components
  analysis}.
\bjournal{Ann. Statist.}
\bvolume{29}
\bpages{295--327}.
\bdoi{10.1214/aos/1009210544}
\bmrnumber{1863961 (2002i:62115)}
\end{barticle}
\endbibitem

\bibitem[\protect\citeauthoryear{Jolliffe}{2002}]{jolliffe2005principal}
\begin{bbook}[author]
\bauthor{\bsnm{Jolliffe},~\bfnm{I.~T.}\binits{I.~T.}}
(\byear{2002}).
\btitle{Principal component analysis},
\bedition{second} ed.
\bseries{Springer Series in Statistics}.
\bpublisher{Springer-Verlag, New York}.
\bmrnumber{2036084 (2004k:62010)}
\end{bbook}
\endbibitem

\bibitem[\protect\citeauthoryear{Josse and Husson}{2012}]{josse2012selecting}
\begin{barticle}[author]
\bauthor{\bsnm{Josse},~\bfnm{Julie}\binits{J.}} \AND
  \bauthor{\bsnm{Husson},~\bfnm{Fran{\c{c}}ois}\binits{F.}}
(\byear{2012}).
\btitle{Selecting the number of components in principal component analysis
  using cross-validation approximations}.
\bjournal{Comput. Statist. Data Anal.}
\bvolume{56}
\bpages{1869--1879}.
\bdoi{10.1016/j.csda.2011.11.012}
\bmrnumber{2892383}
\end{barticle}
\endbibitem

\bibitem[\protect\citeauthoryear{Kritchman and
  Nadler}{2008}]{kritchman2008determining}
\begin{barticle}[author]
\bauthor{\bsnm{Kritchman},~\bfnm{Shira}\binits{S.}} \AND
  \bauthor{\bsnm{Nadler},~\bfnm{Boaz}\binits{B.}}
(\byear{2008}).
\btitle{Determining the number of components in a factor model from limited
  noisy data}.
\bjournal{Chemometrics and Intelligent Laboratory Systems}
\bvolume{94}
\bpages{19--32}.
\end{barticle}
\endbibitem

\bibitem[\protect\citeauthoryear{Mardia, Kent and
  Bibby}{1979}]{mardia1979multivariate}
\begin{bbook}[author]
\bauthor{\bsnm{Mardia},~\bfnm{Kantilal~Varichand}\binits{K.~V.}},
  \bauthor{\bsnm{Kent},~\bfnm{John~T.}\binits{J.~T.}} \AND
  \bauthor{\bsnm{Bibby},~\bfnm{John~M.}\binits{J.~M.}}
(\byear{1979}).
\btitle{Multivariate analysis}.
\bpublisher{Academic Press [Harcourt Brace Jovanovich, Publishers], London-New
  York-Toronto, Ont.}
\bnote{Probability and Mathematical Statistics: A Series of Monographs and
  Textbooks}.
\bmrnumber{560319 (81h:62003)}
\end{bbook}
\endbibitem

\bibitem[\protect\citeauthoryear{Mazumder, Hastie and
  Tibshirani}{2010}]{mazumder2010spectral}
\begin{barticle}[author]
\bauthor{\bsnm{Mazumder},~\bfnm{Rahul}\binits{R.}},
  \bauthor{\bsnm{Hastie},~\bfnm{Trevor}\binits{T.}} \AND
  \bauthor{\bsnm{Tibshirani},~\bfnm{Robert}\binits{R.}}
(\byear{2010}).
\btitle{Spectral regularization algorithms for learning large incomplete
  matrices}.
\bjournal{J. Mach. Learn. Res.}
\bvolume{11}
\bpages{2287--2322}.
\bmrnumber{2719857 (2011m:62184)}
\end{barticle}
\endbibitem

\bibitem[\protect\citeauthoryear{Muirhead}{1982}]{muirhead}
\begin{bbook}[author]
\bauthor{\bsnm{Muirhead},~\bfnm{Robb~J.}\binits{R.~J.}}
(\byear{1982}).
\btitle{Aspects of multivariate statistical theory}.
\bpublisher{John Wiley \& Sons, Inc., New York}
\bnote{Wiley Series in Probability and Mathematical Statistics}.
\bmrnumber{652932 (84c:62073)}
\end{bbook}
\endbibitem

\bibitem[\protect\citeauthoryear{Nadler}{2008}]{nadler2008finite}
\begin{barticle}[author]
\bauthor{\bsnm{Nadler},~\bfnm{Boaz}\binits{B.}}
(\byear{2008}).
\btitle{Finite sample approximation results for principal component analysis: a
  matrix perturbation approach}.
\bjournal{Ann. Statist.}
\bvolume{36}
\bpages{2791--2817}.
\bdoi{10.1214/08-AOS618}
\bmrnumber{2485013 (2010g:62190)}
\end{barticle}
\endbibitem

\bibitem[\protect\citeauthoryear{Reid, Tibshirani and
  Friedman}{2013}]{reid2013study}
\begin{barticle}[author]
\bauthor{\bsnm{Reid},~\bfnm{Stephen}\binits{S.}},
  \bauthor{\bsnm{Tibshirani},~\bfnm{Robert}\binits{R.}} \AND
  \bauthor{\bsnm{Friedman},~\bfnm{Jerome}\binits{J.}}
(\byear{2013}).
\btitle{A study of error variance estimation in lasso regression}.
\bnote{Preprint. Available at
  \href{http://arxiv-web3.library.cornell.edu/abs/1311.5274}{arXiv:1311.5274}.}
\end{barticle}
\endbibitem

\bibitem[\protect\citeauthoryear{Taylor, Loftus and {Tibshirani,
  Ryan}}{2013}]{taylor2013tests}
\begin{barticle}[author]
\bauthor{\bsnm{Taylor},~\bfnm{Jonathan}\binits{J.}},
  \bauthor{\bsnm{Loftus},~\bfnm{Joshua}\binits{J.}} \AND
  \bauthor{\bsnm{{Tibshirani, Ryan}}}
(\byear{2013}).
\btitle{Tests in adaptive regression via the Kac-Rice formula}.
\bnote{Preprint. Available at
  \href{http://arxiv.org/abs/1308.3020}{arXiv:1308.3020}.}
\end{barticle}
\endbibitem

\bibitem[\protect\citeauthoryear{Tibshirani}{1996}]{tibshirani1996regression}
\begin{barticle}[author]
\bauthor{\bsnm{Tibshirani},~\bfnm{Robert}\binits{R.}}
(\byear{1996}).
\btitle{Regression shrinkage and selection via the lasso}.
\bjournal{J. Roy. Statist. Soc. Ser. B}
\bvolume{58}
\bpages{267--288}.
\bmrnumber{1379242 (96j:62134)}
\end{barticle}
\endbibitem

\end{thebibliography}
\nocite{*}
\end{document}